\documentclass[11pt]{article}

\usepackage[margin=1in]{geometry}
\usepackage{amsmath,amssymb,amsthm,mathtools}
\usepackage{enumitem}
\usepackage{aliascnt}
\usepackage{xcolor}
\usepackage{hyperref}
\usepackage{cleveref}
\usepackage{microtype}

\newcommand{\rev}[1]{#1}
\newcommand{\Dpar}[1]{D^{\parallel}(#1)}
\newcommand{\EA}{\ensuremath{\mathrm{EA}}}

\newtheorem{Thm}{Theorem}[section]
\newtheorem*{Thm*}{Theorem}

\newaliascnt{Prop}{Thm}
\newtheorem{Prop}[Prop]{Proposition}
\aliascntresetthe{Prop}

\newaliascnt{Lem}{Thm}
\newtheorem{Lem}[Lem]{Lemma}
\aliascntresetthe{Lem}

\newaliascnt{Cor}{Thm}
\newtheorem{Cor}[Cor]{Corollary}
\aliascntresetthe{Cor}

\newaliascnt{Fact}{Thm}
\newtheorem{Fact}[Fact]{Fact}
\aliascntresetthe{Fact}

\newtheorem{Def}{Definition}
\newtheorem*{Con*}{Condition}

\crefname{Thm}{Theorem}{Theorems}
\crefname{Prop}{Proposition}{Propositions}
\crefname{Lem}{Lemma}{Lemmas}
\crefname{Fact}{Fact}{Facts}
\crefname{Def}{Definition}{Definitions}
\crefname{Cor}{Corollary}{Corollaries}
\crefname{Con}{Condition}{Conditions}
\crefname{equation}{Equation}{Equations}
\crefname{section}{Section}{Sections}
\crefname{appendix}{Appendix}{Appendices}

\NewCommandCopy{\RSMPcrefOriginal}{\cref}
\NewCommandCopy{\RSMPCrefOriginal}{\Cref}
\RenewDocumentCommand{\cref}{s m}{%
  \IfBlankTF{#2}
    {\IfBooleanTF{#1}{\ref*{}}{\ref{}}}
    {\IfBooleanTF{#1}{\RSMPcrefOriginal*{#2}}{\RSMPcrefOriginal{#2}}}%
}
\RenewDocumentCommand{\Cref}{s m}{%
  \IfBlankTF{#2}
    {\IfBooleanTF{#1}{\ref*{}}{\ref{}}}
    {\IfBooleanTF{#1}{\RSMPCrefOriginal*{#2}}{\RSMPCrefOriginal{#2}}}%
}

\newlist{conditions}{enumerate}{1}
\setlist[conditions,1]{
  label=\textup{(C.\Roman*)},
  ref=C.\Roman*,
  leftmargin=*
}

\crefformat{conditionsi}{#2(#1)#3}
\Crefformat{conditionsi}{#2(#1)#3}

\newcommand{\SMP}{\mathsf{SMP}}
\newcommand{\Bset}{\{0,1\}}
\newcommand{\Ext}{\operatorname{Ext}}
\newcommand{\ER}{\overline{R}}
\newcommand{\EC}{\overline{C}}
\newcommand{\Qbar}{\overline{Q}}
\newcommand{\Mopt}{\mathsf{opt}}
\newcommand{\E}{\mathbf{E}}
\newcommand{\Cost}{\mathrm{Cost}}
\newcommand{\Prb}{\mathbb{P}}
\newcommand{\Rel}{\mathcal R}
\newcommand{\AvSP}{\mathsf{Succ}}
\newcommand{\Rg}{\operatorname{Row}(g)}
\newcommand{\Cg}{\operatorname{Col}(g)}
\newcommand{\Rf}{\operatorname{Row}(f)}

\newcommand{\Xg}{\mathcal{X}_g}

\newcommand{\Unif}{\operatorname{Unif}}
\newcommand{\UnifXYng}{\Unif(\mathcal X^n_g \times \mathcal Y^n_g)}
\newcommand{\UnifXYnf}{\Unif(\mathcal X^n_f \times \mathcal Y^n_f)}
\DeclareMathOperator{\Tr}{Tr}

\title{Exact Asymptotic Rates and an Exponential Strong Converse\\
for quantum SMP and One-Way Communication}
\author{Daiki Suruga\\IQC, University of Waterloo}
\date{}

\begin{document}
\maketitle

\begin{abstract}
	Computing many instances jointly can reduce the communication required per
instance. We ask whether such savings are possible in the
simultaneous-message-passing (SMP) model and how they depend on quantum
messages and shared resources.
For every finite total function $f$ and every fixed error
$0\le\varepsilon<1$, we prove that the optimal worst-case communication
per instance for $f^n$ converges to
\[
 D^\parallel(f)=\log|\operatorname{Row}(f)|
                 +\log|\operatorname{Col}(f)|,
\]
where $\operatorname{Row}(f)$ and $\operatorname{Col}(f)$ are the sets of
distinct rows and columns of the table $(f(x,y))_{x,y}$.
This characterization holds for classical and quantum SMP without shared
entanglement, with or without shared randomness.
Thus, directly transmitting the row and column indices is asymptotically
optimal, even when bounded error, joint computation, and quantum messages
are allowed.
In particular, the exponential advantage achieved by quantum
fingerprinting for a single instance of equality without shared
randomness disappears in the asymptotic rate.

The same threshold governs an exponential strong converse: every fixed
(worst-case) communication rate below it forces the probability of computing all
$n$ outputs correctly to decay as $2^{-\Omega(n)}$, even on average
under a fixed product input distribution.
With shared randomness, we also determine the optimal expected
communication rate, which is $(1-\varepsilon)D^\parallel(f)$.
Entanglement shared between each sender and the referee halves both
rates, and arbitrary tripartite entanglement gives no further reduction.
The results extend to a class of total relations and yield corresponding
one-way characterizations, with $D^\parallel(f)$ replaced by
$\log|\operatorname{Row}(f)|$.

Our proof reduces SMP to one-way communication and relates successful
computation to sampling based on quantum min-entropy; we also give a simpler
proof of the key lemma using the pretty-good measurement.
\end{abstract}

\section{Introduction}
\label{sec:introduction}

Communication complexity studies the communication required to compute a
function whose inputs are distributed among different parties.
In the \emph{simultaneous-message-passing (SMP) model}~\cite{Yao79},
Alice and Bob receive inputs $x$ and $y$, respectively, and each sends a
message to a referee, who outputs $f(x,y)$.
The senders cannot communicate with one another, and the cost is the sum
of their message lengths.
This simple restriction makes SMP a useful setting for studying the power
of randomness~\cite{NS96,BK97}, quantum messages~\cite{BCWdW01,GKdW06,GGJL25},
and shared entanglement~\cite{GKRdW06,Gav08,AG23}.
The model also has connections to nonlocality~\cite{BCMdW10},
privacy~\cite{AHMS20,KN21}, and computational complexity~\cite{GRT22}.

These resources can substantially reduce the communication needed for a
single instance, while it is less clear how much they help when the senders
must \emph{solve many instances jointly}.
For a fixed finite total function
$f:\mathcal X\times\mathcal Y\to\mathcal Z$, we consider the task
\[
 f^n(x^n,y^n):=\bigl(f(x_1,y_1),\ldots,f(x_n,y_n)\bigr),
\]
requiring a correct output with probability at least $1-\varepsilon$ on
every input.
Each sender may use a message that depends on all $n$ inputs in this case, so a
protocol need not handle the instances separately.
We ask for the optimal communication per instance as $n\to\infty$
with $f$ fixed, and investigate carefully how the answer depends on error, quantum messages,
and shared resources.

\emph{Equality} shows why the answer need not follow from the complexity
of one instance.
For $x,y\in\{0,1\}^m$, let $\operatorname{EQ}_m(x,y)=1$ if and only if
$x=y$.
At fixed constant error, classical SMP with public randomness uses only
$O(1)$ bits~\cite{NS96}, while quantum fingerprinting uses $O(\log m)$
qubits without shared randomness or entanglement~\cite{BCWdW01}.
Without shared randomness, the latter gives an exponential advantage over
classical SMP~\cite{NS96,BCWdW01}.
These protocols are also much cheaper than classical zero-error SMP,
which requires $2m$ bits.
When $m$ is fixed and $n$ grows, do these savings persist in the optimal
rate for $\operatorname{EQ}_m^n$?

\emph{Our characterization} answers this question for every finite total
function.
For classical SMP and quantum SMP without shared entanglement, the optimal
worst-case communication per instance converges, for every fixed
$0\le\varepsilon<1$, to
\[
 D^\parallel(f):=\log|\operatorname{Row}(f)|
                 +\log|\operatorname{Col}(f)|,
\]
where $\operatorname{Row}(f)$ and $\operatorname{Col}(f)$ are the sets of
distinct rows $f(x,\cdot)$ and columns $f(\cdot,y)$.
All logarithms are to base two.
This rate is achieved with zero error by encoding the sequence of row
indices and the sequence of column indices, using
$nD^\parallel(f)+O(1)$ bits in total.
Our lower bound shows that bounded error, shared randomness, quantum
messages, and joint computation do not improve on this rate.
Even shared entanglement does not change the answer drastically: entanglement
between each sender and the referee only halves the rate, as achieved by
superdense coding, and arbitrary tripartite entanglement gives no further
reduction.
Thus, across the models considered here, the asymptotic quantum advantage
is at most a factor of two.

For equality, the rate is therefore $2m$ bits or qubits without
shared entanglement, and is $m$ bits with shared entanglement.
The exponential separation between classical and quantum SMP for one
instance disappears in the asymptotic rate.

\emph{The exact rate also determines an exponential transition in success
probability.}
Let $R_*$ denote the corresponding worst-case threshold:
$D^\parallel(f)$ without entanglement and $D^\parallel(f)/2$ with
sender--referee or arbitrary tripartite entanglement.
We prove that, for every fixed $\Delta>0$,
\[
 \operatorname{cost}(\Pi_n)\le n(R_*-\Delta)
 \quad\Longrightarrow\quad
 \min_{x^n,y^n}
 \Pr\!\left[\Pi_n(x^n,y^n)=f^n(x^n,y^n)\right]
 \le 2^{-\Omega(n)}.
\]
In fact, this bound holds for average success probability under a fixed
product input distribution.
Together with the zero-error protocols at rate $R_*+o(1)$, this bound
establishes an \emph{exponential strong converse}: every fixed saving
per instance below the optimal rate makes success exponentially unlikely.
This formulation connects our question to the study of strong converses
in information theory~\cite{Wol57,Ari73,ON99,LWD16,TW19}.
For expected communication, shared randomness permits a different
tradeoff. Running a zero-error protocol with probability $1-\varepsilon$
and otherwise sending no messages achieves rate $(1-\varepsilon)R_*$;
we prove that this rate is optimal as well.

\emph{Direct sum and direct product theorems} provide the closest
comparisons within communication complexity because they also ask how
the cost and success probability change when many instances must be
solved together.
A direct sum theorem seeks a lower bound of $\Omega(nC_0)$, where $C_0$
is the bounded-error complexity of one instance in the relevant model;
such questions have been studied in SMP and other communication
models~\cite{CSWY01,JRS03d,JK09,BR11,Bra15}.
A strong direct product theorem further establishes exponential decay
below a specified communication bound.
A typical form is
\[
 \operatorname{cost}(\Pi_n)\le cnL(f)
 \quad\Longrightarrow\quad
 \Pr[\Pi_n\text{ succeeds}]\le 2^{-\Omega(n)},
\]
where $c>0$ is sufficiently small and $L(f)$ is the complexity measure
appearing in the theorem~\cite{Kla10,She11SDP,JK21,JK22}.
These theorems do not determine the optimal asymptotic rate or give
exponential decay at every rate below it.
Equality makes the distinction concrete: its public-coin bounded-error
complexity is $O(1)$, whereas its classical SMP threshold is $2m$.
Identifying the threshold therefore requires more than improving a
constant factor in a bound expressed through single-instance
bounded-error complexity.

We next state our results for a class of total relations containing all
total functions, then explain the proof and compare the relevant prior
results in more detail.

\subsection{Main results}
\label{subsec:main-results}

Our extension to relations uses two assumptions.
The first allows us to derive a lower bound from an underlying function;
the second ensures that a matching protocol exists.
Let $\Rel$ be a finite total relation, and suppose that a fixed
postprocessing map recovers a function value $g(x,y)$ from every valid
output of $\Rel(x,y)$.
This is Condition~\cref{cond:1}, stated formally in \cref{sec:setup}.
Any protocol for $\Rel$ then gives a protocol for $g$, allowing a lower
bound for $g$ to apply to $\Rel$.
Write
\[
 \Dpar{g}:=\log|\Rg|+\log|\Cg|.
\]
For the matching upper bound, Condition~\cref{cond:2} ensures that the row
$g(x,\cdot)$ and column $g(\cdot,y)$ determine a valid output of $\Rel(x,y)$.
Every total function satisfies both conditions with $g=f$.

Shared randomness is allowed throughout this subsection.
We use the notation SR for quantum SMP without shared entanglement, BE for independent
bipartite entanglement between each sender and the referee, and TE for
arbitrary tripartite entanglement.
Communication is measured in bits for classical protocols and qubits for
quantum protocols; in both cases, the cost is the sum of the message
lengths.
Success for $\Rel^n$ means producing valid outputs in all $n$ coordinates.

\paragraph{First result: Exponential strong converse.}
Our first result establishes an exponential strong converse for
relations satisfying the first assumption.
\begin{Thm}[Informal; see \cref{thm:general-relation-strong}]
\label{thm:intro-strong-converse}
Let $\Rel$ and $g$ satisfy Condition~\cref{cond:1}, and fix $\Delta>0$.
Every SMP protocol for $\Rel^n$ with worst-case communication at most
\[
 \begin{cases}
 n\bigl(\Dpar{g}-\Delta\bigr),
   &\text{classical or quantum without entanglement (SR)},\\[2pt]
 n\bigl(\Dpar{g}/2-\Delta\bigr),
   &\text{quantum with prior entanglement (BE, TE)},
 \end{cases}
\]
has average success probability at most $2^{-\kappa n}$ for all sufficiently
large $n$, where $\kappa>0$ depends only on $g$, $\Delta$, and the model.
\end{Thm}
The average is taken over
$(X^n,Y^n)\sim\Unif(\mathcal X_g^n\times\mathcal Y_g^n)$.
Here $\mathcal X_g$ contains one fixed input for each distinct row of $g$,
and $\mathcal Y_g$ contains one fixed input for each distinct column.
Thus the input distribution depends only on $g$, and the same distribution
works for every protocol.
Since some input has success probability no larger than this average,
the theorem also bounds worst-case success probability.

\paragraph{Second result: Exact asymptotic rates.}
Under Condition~\cref{cond:2}, sending the sequences of row and column
indices solves $\Rel^n$ with zero error using $n\Dpar{g}+O(1)$ bits.
Superdense coding halves this cost when each sender shares entanglement
with the referee.
These protocols match the strong converse for worst-case communication;
a truncation argument also gives the matching lower bounds for expected
communication, as shown in \cref{sec:tight-characterization}.

Write $\ER^{\SMP}_\varepsilon(\Rel^n)$ and
$R^{\SMP,\mathrm{wc}}_\varepsilon(\Rel^n)$ for the optimal expected and
worst-case classical costs with worst-case error at most $\varepsilon$.
Expected cost means the maximum, over inputs, of the expected number of
bits communicated.
The corresponding quantum costs are $\Qbar^{\mathrm{res}}_\varepsilon(\Rel^n)$
and $Q^{\mathrm{res},\mathrm{wc}}_\varepsilon(\Rel^n)$, where
$\mathrm{res}\in\{\mathrm{SR},\mathrm{BE},\mathrm{TE}\}$.

\noindent\begin{minipage}{\linewidth}
\begin{Thm}[Informal; see \cref{thm:matching-relation-rates}]
\label{thm:intro-matching-rates}
Let $\Rel$ and $g$ satisfy Conditions~\cref{cond:1} and \cref{cond:2}.
For every fixed $0\le\varepsilon<1$, the following limits hold:
\begin{center}
\renewcommand{\arraystretch}{2.1}
\setlength{\tabcolsep}{5pt}
\begin{tabular}{@{}lcc@{}}
\hline
Model & Expected rate & Worst-case rate \\
\hline 
\noalign{\vskip 6pt}
Classical
 & $\displaystyle\lim_{n\to\infty}\frac{\ER^{\SMP}_\varepsilon(\Rel^n)}{n}
       =(1-\varepsilon)\Dpar{g}$
 & $\displaystyle\lim_{n\to\infty}\frac{R^{\SMP,\mathrm{wc}}_\varepsilon(\Rel^n)}{n}
       =\Dpar{g}$ \\[6pt]
Quantum (SR)
 & $\displaystyle\lim_{n\to\infty}\frac{\Qbar^{\mathrm{SR}}_\varepsilon(\Rel^n)}{n}
       =(1-\varepsilon)\Dpar{g}$
 & $\displaystyle\lim_{n\to\infty}\frac{Q^{\mathrm{SR},\mathrm{wc}}_\varepsilon(\Rel^n)}{n}
       =\Dpar{g}$ \\[6pt]
Quantum (BE, TE)
 & $\displaystyle\lim_{n\to\infty}\frac{\Qbar^{\mathrm{res}}_\varepsilon(\Rel^n)}{n}
       =\frac{1-\varepsilon}{2}\Dpar{g}$
 & $\displaystyle\lim_{n\to\infty}\frac{Q^{\mathrm{res},\mathrm{wc}}_\varepsilon(\Rel^n)}{n}
       =\frac{\Dpar{g}}{2}$ \\[6pt]
\hline
\end{tabular}
\end{center}
In the last row, $\mathrm{res}\in\{\mathrm{BE},\mathrm{TE}\}$.
\end{Thm}
\end{minipage}

The expected rate follows the same threshold, with the factor
$1-\varepsilon$ accounting exactly for the option of not running the
zero-error protocol.
For worst-case communication, any fixed positive success probability
requires the full rate.
These worst-case characterizations also hold without shared randomness:
the lower bounds allow it, while the matching zero-error protocols do
not use it.
Finally, the identical BE and TE rates show that arbitrary tripartite
entanglement offers no asymptotic improvement over independent
sender--referee entanglement.

\paragraph{Consequences for one-way communication.}
When Alice sends a message directly to Bob, the same strong converse and
exact rates hold with $\Dpar{g}$ replaced by
\[
 D_{\to}(g):=\log|\Rg|;
\]
see \cref{cor:one-way-relation-strong,cor:one-way-relation-rates}.
Only the row term remains because Bob already knows his input and
produces the output himself.

\paragraph{Examples and limitations.}
Beyond total functions, the assumptions cover several relations with
multiple valid outputs.
For example, consider the relation that outputs $\mathsf{disj}$ when
two subsets of $[m]$ are disjoint and otherwise outputs an element of
their intersection:
\[
 \Rel_{\cap}(x,y):=
 \begin{cases}
  \{\mathsf{disj}\},&x\cap y=\varnothing,\\
  x\cap y,&x\cap y\ne\varnothing.
 \end{cases}
\]
The results also apply to relations of the form
$\Rel_h(x,y)=\{(a,b)\in G^2:a+b=h(x,y)\}$, where $G$ is a finite
abelian group, as in linear games~\cite{RAM16}.
The conditions are restrictive, however: the magic-square
relation~\cite{Mer90,Per90} falls outside the characterization.
\Cref{sec:applications} gives the applications and explains these
limitations.

\subsection{Proof technique}
\label{subsec:proof-technique}

The main task is to prove the exponential strong converse.
Once it is established, matching protocols and truncation yield the exact
rates; see \cref{sec:tight-characterization}.
We explain the argument first for fixed-length quantum messages without
shared entanglement or shared randomness, then describe how to incorporate
these resources and variable message lengths.

\paragraph{Reduction to one-way communication.}
An SMP protocol can save communication only if at least one sender's
message is short.
More precisely, if its message lengths satisfy
$q_A+q_B\le n(\Dpar{g}-\Delta)$, then
\[
 q_A\le n\bigl(\log|\Rg|-\Delta/2\bigr)
 \quad\text{or}\quad
 q_B\le n\bigl(\log|\Cg|-\Delta/2\bigr).
\]
In the first case, group Bob and the referee into one receiver.
This receiver knows $y^n$, so he can prepare Bob's message and perform
the referee's measurement after receiving Alice's message.
The resulting one-way protocol has the same success probability as the
SMP protocol but communicates only $q_A$ qubits.
The second case is symmetric; see \cref{subsec:one-sender-views}.
Condition~\cref{cond:1} then converts successful computation of $\Rel^n$
into successful computation of $g^n$.
It therefore suffices to prove a one-way strong converse for $g^n$.

\paragraph{The one-way statement.}
Let Alice's $q_A$-qubit message of a one-way protocol $\Pi$ be $Q$, and suppose that, for a fixed
$\Delta>0$,
\begin{equation}
 q_A\le n\bigl(\log|\Rg|-\Delta\bigr).
 \label{eq:intro-one-way-gap}
\end{equation}
Take $X^n$ and $Y^n$ independently and uniformly from
$\mathcal X_g^n$ and $\mathcal Y_g^n$.
A short message makes Alice's entire input hard to recover, but this
alone does not bound the probability of computing $g^n(X^n,Y^n)$:
the output may reveal much less than $X^n$.
To bridge this gap, we encode the row truth tables determined by $X^n$
as a binary string $V^n=V_n(X^n)$ of length $N=\Theta(n)$.
This encoding is injective on $\mathcal X_g^n$, and its entries can be
recovered from suitable function values.
For a random set $T\subseteq[N]$, write $V_T$ for the sampled bits.
The argument proceeds through three implications:
\begin{center}
\begin{minipage}{0.94\linewidth}
\begin{enumerate}[label=\textbf{Step \arabic*:},leftmargin=*]
 \item The communication deficit gives
 $H_{\min}(X^n\mid Q)\ge\Delta n$.
 \item This entropy bound makes a random substring hard to recover:
 $p_{\mathrm{guess}}(V_T\mid TQ)\le 2^{-\Omega(|T|)}$.
 \item A protocol that computes $g^n$ yields a decoder for $V_T$,
 implying exponentially small success probability for the protocol.
\end{enumerate}
\end{minipage}
\end{center}

\paragraph{Step 1: Short communication leaves linear min-entropy.}
Conditional min-entropy measures the difficulty of recovering a classical
input from a quantum register:
\[
 H_{\min}(X^n\mid Q)=-\log p_{\mathrm{guess}}(X^n\mid Q),
\]
where $p_{\mathrm{guess}}$ is the optimal success probability over all
measurements on $Q$.
Since $X^n$ is uniform and $Q$ contains $q_A$ qubits,
we show in~\cref{lem:sr-entropy} that
\[
 H_{\min}(X^n\mid Q)
 \ge n\log|\Rg|-q_A
 \ge\Delta n.
\]
The injective encoding preserves this uncertainty, so the same lower
bound holds for $H_{\min}(V^n\mid Q)$.
This supplies the hypothesis needed for sampling.

\paragraph{Step 2: A random substring remains hard to guess.}
A sampling theorem from Ref.~\cite{Wul11} converts the conditional
min-entropy bound from Step~1 into a bound on the probability of guessing
randomly selected bits.
We state the required form in \cref{lem:wullschleger} and give its proof
in \cref{app:wullschleger}.
For sufficiently small constant $\eta>0$, let $T$ be a uniformly random
$k$-element subset of $[N]$, independent of $X^nQ$, where
$k=\lfloor\eta n\rfloor$.
The sampling bound gives
\begin{equation}
 p_{\mathrm{guess}}(V_T\mid TQ)\le 2^{-\Omega(k)}.
 \label{eq:intro_step2}
\end{equation}
Thus, even after learning which positions $T$ were sampled, the receiver
has exponentially small probability of recovering all their values.
To use this bound, we must relate the sampled bits to outputs that a
single execution of the protocol can provide.

\paragraph{Step 3: Successful computation recovers the substring.}
For a suitable set $\mathcal G$ of samples $T$, we can choose Bob's input
so that the protocol's outputs determine every bit in $V_T$.
The construction in \cref{lem:functional-sampling} also preserves the
original input distribution after averaging over these samples.
It follows that the resulting decoder for $V_T$ based on the one-way protocol $\Pi$ succeeds with probability at least
\[
 \Pr[T\in\mathcal G]\,\Pr[\Pi\text{ succeeds}].
\]
The optimal guessing probability is at least this large.
We therefore obtain
\[
 \Pr[\Pi\text{ succeeds}]
 \le\frac{p_{\mathrm{guess}}(V_T\mid TQ)}{\Pr[T\in\mathcal G]}.
\]
The remaining issue is the probability of obtaining a suitable sample.
We show that $\Pr[T\in\mathcal G]\ge 2^{-O(k^2/n)}$, and hence
\[
 \Pr[\Pi\text{ succeeds}]
 \le 2^{-\Omega(k)+O(k^2/n)}.
\]
For $k=\lfloor\eta n\rfloor$, the negative term has order $\eta n$,
whereas the loss has order $\eta^2n$.
Choosing $\eta$ sufficiently small preserves exponential decay.
\Cref{lem:functional-sampling} formalizes the sampling and decoding
argument, and \cref{prop:one-way-quantum} combines it with Step~1.

\paragraph{Shared resources and variable message lengths.}
With prior entanglement, the receiver's side information includes his
share of the entangled state.
The entropy bound in Step~1 then has communication term $2q_A$ in place
of $q_A$, yielding half the threshold.
Grouping the other sender and the referee into one receiver makes this
argument applicable to both BE and TE.
For shared randomness, we apply the bound for each fixed randomness
value and average; its constants are uniform over those values.
Finally, \cref{lem:fixed-register-reduction} converts variable message
lengths to fixed registers with only $O(1)$ additional qubits, which does
not change the asymptotic threshold or exponential decay.

After the completion of this manuscript, a simpler proof of
\cref{lem:functional-sampling} was found.
It uses the pretty-good measurement directly and avoids introducing the
random sample $T$.
We include this proof in \cref{app:simplified-proof}, while retaining the
original sampling proof in the main text.

\subsection{Related work}

The comparisons below distinguish the range of relations covered, the
shared resources allowed, and the communication bound at which each
converse applies. These distinctions explain both the scope of the earlier
results and the additional conclusions provided by our characterization.

\emph{Direct sums in SMP.}
Ref.~\cite{CSWY01} introduced an information-theoretic approach to direct
sum problems in SMP.
Ref.~\cite{JK09} proves classical and quantum direct sum theorems for
arbitrary relations when shared resources are available separately
between each sender and the referee, with no shared resource between
the senders.
For fixed error parameters, these results give lower bounds of order
$\Omega(nC_0)$, where $C_0$ is the corresponding single-instance
bounded-error complexity.
They establish linear growth but do not identify its exact coefficient
or determine the success probability below the optimal rate.
Our results answer both questions for the stated class of relations,
allowing randomness shared by all parties and arbitrary tripartite
entanglement in the quantum setting.
The earlier direct sum results apply to relations outside this class.

\emph{Direct product theorems.}
Refs.~\cite{Jai15,JK21} prove strong direct product theorems for classical
public-coin and entanglement-assisted quantum one-way communication,
respectively, for arbitrary relations.
For fixed error parameters, their bounds give exponential decay when
communication is at most $n(cC_0-b)$, for suitable constants $c>0$ and
$b\ge0$, where $C_0$ is the relevant bounded-error complexity of one
instance.
Our one-way threshold is instead $\log|\Rg|$ without entanglement and
half this value with entanglement, with decay at every fixed rate below it.
For example, public-coin one-way equality has bounded-error complexity
$O(1)$, but its classical asymptotic threshold is $m$.
This is a difference in the quantity governing the threshold, in
addition to the range of rates covered by the converse.

Strong direct product results also hold for broader communication
models~\cite{Kla10,She11SDP,BW15,JK22}.
Ref.~\cite{BW15} gives a classical interactive result in terms of
information complexity, while Ref.~\cite{JK22} treats
entanglement-assisted multiparty quantum communication, including SMP
protocols as a special case.
The bounds for these broader models need not match the optimal SMP rate.
Our characterization uses the restriction to simultaneous messages to
obtain an explicit threshold for every total function.

\emph{Exact asymptotic rates.}
There are several results that use information theoretic techniques to characterize asymptotic communication in
classical interactive models under per-instance error
criteria~\cite{BR11,Bra15}, with a quantum counterpart developed in
Ref.~\cite{Tou15}.
Ref.~\cite{AT16} gives an exact characterization of asymptotic expected
communication in classical SMP under a per-instance error criterion,
with independent randomness shared between each sender and the referee.
These results address a different error requirement: each coordinate
must be correct with the prescribed probability, without requiring all
coordinates to be correct simultaneously with that probability.
For the latter requirement, Ref.~\cite{Sur26} characterizes the exact
asymptotic expected cost in classical interactive communication and the
worst-case cost within a constant factor.
To our knowledge, the present work gives the first exact
characterizations of both rates for all total functions in the SMP
models considered here.

\emph{Information-theoretic strong converses.}
Ref.~\cite{TW19} proves strong converses for classical bounded-round
interactive function computation under i.i.d.\ input distributions;
its proofs also yield exponential strong converses.
The threshold is the optimal distributional rate under that model's
communication and output conventions.
Such a result gives a worst-case lower bound because worst-case success
probability is at most distributional success probability.
Adapting the output convention to require an answer only at the receiver
also connects this approach to the classical one-way setting.
Our results treat the stated classical and quantum SMP models and give
explicit thresholds together with matching expected and worst-case
rates.

\paragraph{Open problems.}
The exact SMP characterization raises the question of how far these
results extend to interactive communication.
Beyond the bounded-round classical setting of Ref.~\cite{TW19}, can one
determine the exact asymptotic worst-case communication rates and prove
exponential strong converses for classical and quantum two-way protocols?

Even in SMP, the threshold does not determine how quickly the optimal
success probability decays at a given rate below it.
Determining the optimal exponent would quantify what can still be
achieved with insufficient communication.
For relations, a further question is whether an analogous
characterization is possible without Conditions~\cref{cond:1} and
\cref{cond:2}, which allow us to pass through a function value.

The location of shared entanglement presents a separate question.
Suppose Alice and Bob share entanglement with each other, but neither
shares entanglement with the referee.
Our lower bound for arbitrary tripartite entanglement still applies,
but the superdense-coding protocol attaining it requires
sender--referee entanglement.
What determines the exact rate when only the senders share entanglement,
and does an exponential strong converse hold at that rate?

\paragraph{AI disclosure}
The author contributed (sadly?) only the research ideas and questions, while ChatGPT (5.6/Astra, accessed through ChatGPT Plus) developed and completed all proofs. 
The author subsequently verified every proof step by step and wrote the manuscript with further AI assistance to improve readability. 
The author takes full responsibility for the correctness and presentation of the paper. 
The original AI outputs are archived on \href{https://osf.io/qewp4/overview?view_only=56400ca0ef034fcc8e68d208c9fc3142}{OSF}.

\subsection{Organization}
\begin{description}[
  style=multiline,
  leftmargin=2.5cm,
  labelwidth=2.2cm,
  align=left,
  font=\normalfont,
  nosep
]
 \item[\Cref{sec:setup}]
 Communication models and assumptions on the relations considered.
 \item[\Cref{sec:strong-converse}]
	 Exponential strong converse~(\cref{thm:intro-strong-converse}) for SMP communication.
 \item[\Cref{sec:tight-characterization}]
 Exact asymptotic rates~(\cref{thm:intro-matching-rates}) for worst-case and expected communication.
 \item[\Cref{sec:one-way}]
 Corresponding results for one-way communication.
 \item[\Cref{sec:applications}]
 Applications and limitations.
 \item[\Cref{app:wullschleger}]
 Proof of the min-entropy sampling lemma.
 \item[{\Cref{app:simplified-proof}}]
An alternative proof of \cref{lem:functional-sampling}
 using the pretty-good measurement.
\end{description}

\section{Preliminaries and formal setup}
\label{sec:setup}
All logarithms are to base two.
Throughout the paper, $\mathcal X,\mathcal Y,\mathcal Z,\mathcal W$ are finite nonempty sets.
\subsection{Relations and row--column structure}

A total relation is a map
\[
 \Rel:\mathcal X\times\mathcal Y\longrightarrow 2^{\mathcal Z}\setminus\{\varnothing\}.
\]
Its $n$-fold version is
\[
 \Rel^n(x^n,y^n):=\Rel(x_1,y_1)\times\cdots\times\Rel(x_n,y_n)
\]
where ``$\times$'' denotes the Cartesian product.
Let $g:\mathcal X\times\mathcal Y\to\mathcal W$ be a function.  For $x\in\mathcal X$ and
$y\in\mathcal Y$, write $g_x:=g(x,\cdot), ~ g^y:=g(\cdot,y)$,
and define 
\[
 \operatorname{Row}(g):=\{g_x:x\in\mathcal X\},
 \qquad
 \operatorname{Col}(g):=\{g^y:y\in\mathcal Y\}.
\]
These are the sets of distinct rows and columns of the table of $g$.

Our main results below concern the following conditions: 

\begin{Con*}
Let $\Rel$ be a total relation. Consider the following conditions.

\begin{conditions}
\item\label{cond:1}
There are maps
$g:\mathcal X\times\mathcal Y\to\mathcal W$ and $\phi:\mathcal Z\to\mathcal W$
such that, for every $(x,y)\in\mathcal X\times\mathcal Y$,
\begin{equation}
\Rel(x,y)
\subseteq
\{z\in\mathcal Z:\phi(z)=g(x,y)\}.
\label{eq:relation-structure_0}
\end{equation}

\item\label{cond:2}
There are maps $g, \phi$ satisfying \cref{cond:1}
and an additional map $s:\operatorname{Row}(g)\times\operatorname{Col}(g)\to\mathcal Z$
satisfying
\begin{equation}
s(g_x,g^y)\in\Rel(x,y)
\label{eq:relation-structure}
\end{equation}
for every $(x,y)\in\mathcal X\times\mathcal Y$.
\end{conditions}
\end{Con*}

\cref{cond:1} gives the lower bound in~\cref{thm:general-relation-strong}, while
\cref{cond:2} supplies the matching zero-error protocol used in~\cref{thm:matching-relation-rates}.
A particularly important subclass, that satisfies \cref{cond:2}, is
\begin{equation}
 \Rel(x,y)=\{z\in\mathcal Z:\phi(z)=g(x,y)\}
 \label{eq:function-defined-relation}
\end{equation}
This subclass contains all total functions $f$.\footnote{ Take $g := f$ and $\phi$ being identity.}
Indeed, for a total relation of this form, a map $s$ in \cref{cond:2} is obtained by fixing
one element $t(w)\in\phi^{-1}(w)$ for each $w$ in the range of $g$ and setting
$s(g_x,g^y):=t(g(x,y))$.  This is well-defined because $g(x,y)$ depends only on the pair
$(g_x,g^y)$. Thus the relation $\Rel$ satisfies \cref{cond:2}. 

As shown in \cref{lem:relation-to-function}, any SMP protocol for relations satisfying~\cref{cond:1} is reduced to that for computing $g$:
\begin{Lem}
\label{lem:relation-to-function}
Suppose that \cref{cond:1} holds. For any $n$, any SMP protocol for $\Rel^n$
can be converted into a protocol for $g^n$ with the same messages, shared resource, and
communication.  On every fixed input,
\begin{equation}
 \Prb\bigl[\Pi_\Rel(x^n,y^n)\in\Rel^n(x^n,y^n)\bigr] 
 \le
 \Prb\bigl[\phi^n(\Pi_\Rel(x^n,y^n))=g^n(x^n,y^n)\bigr].
 \label{eq:relation-success-below-function}
\end{equation}
\end{Lem}
\begin{proof}
Keep the encoders, shared resource, and referee's measurement unchanged, and postprocess the
classical output by $\phi^n$.  If $z^n$ is valid for the relation, then
\cref{cond:1} holds in every coordinate, so
$\phi^n(z^n)=g^n(x^n,y^n)$.  The transformation changes no message or quantum message register.
\end{proof}

For any (finite) function $g$, let
\begin{equation}
 \Dpar{g}:=\log |\Rg| +\log |\Cg|
 \label{eq:Dg}
\end{equation}
for SMP communication, and put
\begin{equation}
 D_{\to}(g):=\log |\Rg|.
 \label{eq:Dto}
\end{equation}
for one-way communication from Alice to Bob.

Denote sets $\mathcal{X}_g\subseteq\mathcal X$ and $\mathcal Y_g\subseteq\mathcal Y$ containing one
equivalence class of each distinct row and column of $g$. Therefore $x, x' \in \mathcal{X}_g$ and $x \neq x'$ means $g(x, \cdot) \neq g(x', \cdot)$ and $\Rg = \{g_x : x \in \mathcal{X}_g\}$.

\subsection{SMP protocols and communication costs}

\paragraph{Classical protocols.}
  A public-coin classical SMP protocol $\Pi$ uses shared randomness $R$ shared by
Alice, Bob, and the referee, as well as independent private randomness.  For every fixed $R=r$,
If the lengths of their messages $M_A, M_B$ are
$\ell_A(M_A)$ and $\ell_B(M_B)$, define
\[
 \Cost_{\rev{\mathrm{exp}}}(\Pi)
 :=\max_{x^n,y^n}
 \E_{R,R_A,R_B}
 \bigl[\ell_A(M_A(x^n,R,R_A))+\ell_B(M_B(y^n,R,R_B))\bigr],
\]
and
\[
 \Cost_{\mathrm{wc}}(\Pi)
 :=\max_{x^n,y^n,r,r_A,r_B}
 \bigl(\ell_A(M_A(x^n,r,r_A))+\ell_B(M_B(y^n,r,r_B))\bigr),
\]
where the maximum is over the supports of the random variables.  A protocol solves $\Rel^n$ with
global error at most $\varepsilon$ if, for every $(x^n,y^n)$,
\[
 \Prb\bigl[\Pi(x^n,y^n)\in\Rel^n(x^n,y^n)\bigr]\ge1-\varepsilon.
\]
For a function, success means equality with the complete function-value vector.  For either a
relation or a function $P$, write
\[
 \ER^{\SMP}_\varepsilon(P^n),
 \qquad
 R^{\SMP,\mathrm{wc}}_\varepsilon(P^n),
\]
for the optimal classical expected and worst-case costs, respectively.

\paragraph{Quantum protocols.}
In a quantum SMP protocol, Alice, Bob, and the referee share randomness $R$.  Conditioned on
$R=r$, they also share a state satisfying one of the following resource restrictions:
\begin{enumerate}[label=(\roman*)]
 \item SR: shared randomness and no entanglement;
 \item BE: independent entangled states between each sender and the referee;
 \item TE: arbitrary tripartite entanglement.
\end{enumerate}
For fixed inputs $x,y$ and shared-randomness value $r$, Alice and Bob choose nonnegative integers
$q_A(x,r)$ and $q_B(y,r)$.  Each sender adjoins a local register in a fixed pure state, applies an
input- and $r$-dependent unitary operator, sends the designated $q_A(x,r)$ or $q_B(y,r)$ output
qubits, and retains the remaining output registers. 
The referee knows $r$ and the two received message lengths
and may choose its measurement accordingly.  The message lengths are deterministic once the
input and $r$ are fixed.

Define
\[
 \operatorname{QCost}_{\rev{\mathrm{exp}}}(\Pi)
 :=\max_{x,y}\E_R\bigl[q_A(x,R)+q_B(y,R)\bigr],
 \qquad
 \operatorname{QCost}_{\mathrm{wc}}(\Pi)
 :=\max_{x,y,r}\bigl(q_A(x,r)+q_B(y,r)\bigr),
\]
where the maximum in the worst-case cost is over the support of $R$.  We write
\[
 \Qbar^{\mathrm{res}}_\varepsilon(P^n),\qquad
 Q^{\mathrm{res},\mathrm{wc}}_\varepsilon(P^n),
 \qquad \mathrm{res}\in\{\mathrm{SR},\mathrm{BE},\mathrm{TE}\},
\]
for the optimal expected and worst-case quantum costs.
Finally, put
\begin{equation}
 \alpha_{\mathrm{SR}}:=1,
 \qquad
 \alpha_{\EA}=\alpha_{\mathrm{BE}}=\alpha_{\mathrm{TE}}:=2.
 \label{eq:alpha-res}
\end{equation}
We use $\alpha_{\mathrm{SR}}$ and $\alpha_{\EA}$ for the one-way models as well.

\subsection{One-way protocols}

In a one-way protocol, Alice receives $x$, Bob receives $y$, Alice sends one message to Bob, and
Bob produces the output.  A classical protocol may use public and private randomness, and Alice's
message is a finite binary string.  Its
expected and worst-case costs are defined by
\[
 \Cost^{\to}_{\rev{\mathrm{exp}}}(\Pi)
 :=\max_{x^n}\E[\ell(M(x^n))],
 \qquad
 \Cost^{\to}_{\mathrm{wc}}(\Pi)
 :=\max_{x^n,r,r_A}\ell(M(x^n,r,r_A)).
\]
For quantum one-way protocols we consider SR, meaning shared randomness and no entanglement, and
\EA, meaning prior entanglement between Alice and Bob.  Conditioned on her input $x$ and
the shared-randomness value $R=r$, Alice applies an input- and $r$-dependent unitary operator as
above and sends $q(x,r)$ qubits.  We put
\[
 \operatorname{QCost}^{\to}_{\rev{\mathrm{exp}}}(\Pi)
 :=\max_x\E_R[q(x,R)],
 \qquad
 \operatorname{QCost}^{\to}_{\mathrm{wc}}(\Pi)
 :=\max_{x,r}q(x,r).
\]
Global error is defined pointwise as above, with Bob's output replacing the referee's output.  We
write
\[
 \EC^{\to}_\varepsilon(P^n),\quad
 C^{\to,\mathrm{wc}}_\varepsilon(P^n),\quad
 \Qbar^{\to,\mathrm{res}}_\varepsilon(P^n),\quad
 Q^{\to,\mathrm{res},\mathrm{wc}}_\varepsilon(P^n),
 \qquad \mathrm{res}\in\{\mathrm{SR},\EA\},
\]
for the corresponding optimal costs.

As always a classical message can be transmitted as an orthogonal quantum state:
\begin{Lem}
\label{lem:classical-as-sr-quantum}
For every finite relation or function $P$, every $n$, and every $\varepsilon\ge0$,
\begin{align*}
 \Qbar^{\mathrm{SR}}_\varepsilon(P^n)
 &\le \ER^{\SMP}_\varepsilon(P^n),
 &
 Q^{\mathrm{SR},\mathrm{wc}}_\varepsilon(P^n)
 &\le R^{\SMP,\mathrm{wc}}_\varepsilon(P^n),\\
 \Qbar^{\to,\mathrm{SR}}_\varepsilon(P^n)
 &\le \EC^{\to}_\varepsilon(P^n),
 &
 Q^{\to,\mathrm{SR},\mathrm{wc}}_\varepsilon(P^n)
 &\le C^{\to,\mathrm{wc}}_\varepsilon(P^n).
\end{align*}
\end{Lem}

\begin{proof}
	Only one clarification is made regarding the difference of private randomness and shared randomness;
	In our quantum SMP model we only use shared randomness. As private randomness is perfectly simulated by shared randomness, we obtain the statements.
\end{proof}

\subsection{Conditional min-entropy}
This paper also uses some standard facts on conditional min-entropy. For references, see~\cite{Tom15,KRS09}.
For a classical--quantum state
\[
 \rho_{VQ}=\sum_v p(v)|v\rangle\!\langle v|\otimes\rho_v^Q,
\]
define
\[
 p_{\mathrm{guess}}(V\mid Q)_\rho
 :=\max_{\{E_v\}}\sum_v p(v)\operatorname{Tr}(E_v\rho_v),
 \qquad
 H_{\min}(V\mid Q)_\rho:=-\log p_{\mathrm{guess}}(V\mid Q)_\rho,
\]
where the maximum is over all POVMs on the entire register $Q$. 

If $Q$ is a
composite register $Q_1 Q_2$, the measurement need not be a product or separable measurement.
Nevertheless (classical) adaptive measurements suffices when one of $Q_1$ and $Q_2$ are classical:
\begin{Fact}[]\label{Fact_guess_CCQ}
For a state $\rho_{VQ_1 Q_2}$ defined as 
\[
 \rho_{VQ_1 Q_2}=\sum_{v,q} p(v, q)|v\rangle\!\langle v|_V \otimes|q\rangle\!\langle q|_{Q_1}\otimes \rho_{v,q}^{Q_2},
\]
it holds that 
\begin{equation*}
 p_{\mathrm{guess}}(V\mid Q_1 Q_2)_\rho
 =\max_{\{E_v^q\} \text{~on~} Q_2}\sum_q \Pr(q) \sum_v p(v|q)\operatorname{Tr}(E_v^q\rho^{Q_2}_{v,q})
\end{equation*}
where the maximum is over all adaptive POVMs on $Q_2$ that may depend on $q$.
\end{Fact}
The following fact is also well-known.
\begin{Fact}\label{Fact_cq-state_up}
	For a cq-state $\rho_{VQ}$, $H_\mathrm{min}(V\mid Q)_\rho \le \log |V|$.
\end{Fact}

\section{Exponential strong converse}
\label{sec:strong-converse}
We state and prove the formal version of \cref{thm:intro-strong-converse}.
\begin{Thm}
\label{thm:general-relation-strong}
Let $\Rel$ be a finite total relation satisfying Condition~\cref{cond:1}, and $g, \phi$ be the corresponding functions. Fix $\mathrm{res}\in\{\mathrm{SR},\mathrm{BE},\mathrm{TE}\}$ and $\Delta>0$.  There is a constant
$\kappa=\kappa(g,\Delta,\mathrm{res})>0$ such that, for all sufficiently large $n$, every quantum
SMP protocol $\Pi_n$ with
\begin{equation}
 \operatorname{QCost}_{\mathrm{wc}}(\Pi_n)
 \le
 n\left(\frac{\Dpar{g}}{\alpha_{\mathrm{res}}}-\Delta\right)
 \label{eq:general-relation-strong-cost}
\end{equation}
has, for $(X^n,Y^n)\sim\Unif(\mathcal X_g^n\times\mathcal Y_g^n)$,
\begin{equation}
 \Prb\bigl[
   \Pi_n(X^n,Y^n)\in\Rel^n(X^n,Y^n)
 \bigr]
 \le 2^{-\kappa n}.
 \label{eq:general-relation-strong-success}
\end{equation}
\end{Thm}
By \cref{lem:classical-as-sr-quantum}, in particular, every public-coin classical
protocol of worst-case cost at most $n(\Dpar{g}-\Delta)$ satisfies the same conclusion.

For a function $f$ and the distribution $\UnifXYnf$, we obtain~\cref{cor:function-strong} because any total function satisfies~Condition~\cref{cond:1}.

\begin{Cor}
\label{cor:function-strong}
Let $f:\mathcal X\times\mathcal Y\to\mathcal W$ be a finite total function.  Fix
$\mathrm{res}\in\{\mathrm{SR},\mathrm{BE},\mathrm{TE}\}$ and $\Delta>0$.  There is
$\kappa=\kappa(f,\Delta,\mathrm{res})>0$ such that every quantum SMP protocol $\Pi_n$ satisfying
\begin{equation}
 \operatorname{QCost}_{\mathrm{wc}}(\Pi_n)
 \le n\left(\frac{\Dpar{f}}{\alpha_{\mathrm{res}}}-\Delta\right)
 \label{eq:strong-cost-assumption}
\end{equation}
satisfies, for all sufficiently large $n$ and
$(X^n,Y^n)\sim\UnifXYnf$,
\begin{equation}
 \Prb\bigl[\Pi_n(X^n,Y^n)=f^n(X^n,Y^n)\bigr]
 \le2^{-\kappa n}.
 \label{eq:distributional-strong}
\end{equation}
\end{Cor}

\paragraph{Proof structure for \cref{thm:general-relation-strong}}
We first replace a variable-length protocol by the fixed-register protocol of
\cref{lem:fixed-register-reduction}.  The replacement costs at most one additional qubit per
sender, which is absorbed into the linear gap in
\eqref{eq:general-relation-strong-cost}.
We then establish the strong converse for one-way protocols computing \(g^n\), which yields the corresponding strong converse for $\Rel$ by~\cref{lem:relation-to-function}, as follows.
If the log-dimension of Alice's message is smaller than
$\log|\Rg|/\alpha_{\mathrm{res}}$, Lemmas~\ref{lem:sr-entropy} and~\ref{lem:fixed-marginal}
show that the corresponding one-sender register \(Q_A\) leaves linear conditional min-entropy,
\[
H_{\min}(X^n\mid Q_A)=\Omega(n)
\]
when $X^n$ is $n$-i.i.d. uniform distribution on $\Xg$.
Lemma~\ref{lem:functional-sampling} then implies that the average probability of correctly computing all \(n\) outputs is exponentially small; this yields the strong converse for the one-way model, \cref{prop:one-way-quantum}.  
By symmetry the same conclusion holds when Bob's message is below $\log|\Cg|/\alpha_{\mathrm{res}}$.

To pass from one-way communication to SMP, we consider the state obtained after applying only
Alice's encoder.  A receiver holding Bob's initial share and the referee's share can apply Bob's
encoder for a fixed input and then reproduce the referee's measurement.  The analogous construction
starting from Bob's encoder is symmetric; \cref{subsec:one-sender-views} formalizes these two
one-sender views.
An SMP protocol for $g^n$ involves messages from both Alice and Bob, but the assumed upper bound
\begin{equation*}
\frac{\Dpar{g}}{\alpha_\mathrm{res}}
=
\frac{\log|\Rg| + \log|\Cg|}{\alpha_\mathrm{res}}
\end{equation*}
on their total communication rate guarantees that at least one of the two messages must lie below the corresponding one-way threshold: $\log|\Rg|/\alpha_\mathrm{res}$ or $\log|\Cg|/\alpha_\mathrm{res}$. 
Hence the average success probability of the SMP protocol also decays exponentially, yielding~\cref{thm:general-relation-strong}.

\subsection{Fixed-length message reduction}
We first show any variable length quantum communication is converted to a fixed length communication up to a small additive factor.
\begin{Lem}
\label{lem:fixed-register-reduction}
Every variable-length quantum SMP (resp. one-way) protocol $\Pi$ has a quantum SMP (resp. one-way) protocol $\widehat\Pi$ with
the same shared resource and the same output distribution, whose message-register dimensions are
independent of the inputs once $R=r$ is fixed, and such that
\[
 \operatorname{QCost}_{\mathrm{wc}}(\widehat\Pi)
 \le \operatorname{QCost}_{\mathrm{wc}}(\Pi)+2.
\]
For a variable-length quantum one-way protocol, the corresponding overhead is at most one qubit.
\end{Lem}

\begin{proof}
Fix $R=r$ and put
\[
 L_A(r):=\max_x q_A(x,r),
 \qquad
 L_B(r):=\max_y q_B(y,r).
\]
For every $0\le a\le L_A(r)$, choose a subspace
$\mathcal H^A_a\subseteq\mathsf (\mathbb{C}^2)^{\otimes L_A(r)+1}$ of dimension $2^a$ so that these subspaces are
mutually orthogonal.  This is possible because
\[
 \dim \left(\bigoplus_{a=0}^{L_A(r)}  \mathcal H^A_a\right)
= \sum_{a=0}^{L_A(r)} \dim(\mathcal H^A_a) = \sum_{a=0}^{L_A(r)}2^a=2^{L_A(r)+1}-1<2^{L_A(r)+1}.
\]
On an input $x$ for which $q_A(x,r)=a$, Alice uses a unitary operator with local ancillary
registers to place her original $a$-qubit message in $\mathcal H^A_a$ and sends the resulting
$L_A(r)+1$ qubits.  Bob uses the analogous construction with $L_B(r)+1$ qubits.  The referee first
identifies the two orthogonal length subspaces by the identity projective measurements $P^A_a$ on $\mathcal H^A_a$ and then applies the original measurement for those
lengths. 
Note that by definition the referee know their message length. 
Due to this construction of the fixed length protocol, the output distribution and shared resource are unchanged.

Since the inputs of the two senders range independently,
\[
 L_A(r)+L_B(r)
 =\max_{x,y}\bigl(q_A(x,r)+q_B(y,r)\bigr).
\]
Taking the maximum over $r$ proves the SMP bound.  The same construction with only Alice's
message proves the one-way statement.
\end{proof}

\subsection{One-sender views of an SMP protocol}\label{subsec:one-sender-views}

Conditioned on $R=r$, write the initial shared state as
\[
 \sigma_r^{A_0B_0C_0},
\]
where some registers are trivial in SR and the appropriate tensor-product restriction is imposed in
BE.  Include the pure local ancillary registers in $A_0$ and $B_0$, and let the
encoding unitary operators be
\[
 U^A_{x,r}:A_0\longrightarrow M_AA_1,
 \qquad
 U^B_{y,r}:B_0\longrightarrow M_BB_1.
\]
Define $\rev{Q_A}=M_AB_0C_0$ in the state obtained after applying only Alice's unitary and tracing out
$A_1$.  Define $\rev{Q_B}=M_BA_0C_0$ in the separate state obtained after applying only Bob's unitary and
tracing out $B_1$.  These systems are auxiliary: they need not coexist in one execution.
The fixed-length reduction (\cref{lem:fixed-register-reduction}) ensures that $M_A$ and $M_B$ have input-independent dimensions in
this subsection.

\begin{Lem}
\label{lem:one-sender-simulation}
Fix a public-randomness value $R=r$.
\begin{enumerate}[label=(\roman*),leftmargin=2.5em]
 \item For every fixed $y$, there is a measurement on $\rev{Q_A}$ whose outcome distribution on input
 $x$ equals the output distribution of the original protocol on $(x,y)$.
 \item For every fixed $x$, there is a measurement on $\rev{Q_B}$ whose outcome distribution on input
 $y$ equals the output distribution of the original protocol on $(x,y)$.
\end{enumerate}
\end{Lem}

\begin{proof}
For (i), start from the Alice-only state on $M_AB_0C_0$.  Apply $U^B_{y,r}$ to $B_0$,
discard $B_1$, and perform the original referee measurement on $M_AM_BC_0$.  The measured state is
\[
 \operatorname{Tr}_{A_1B_1}
 \left[
  (U^A_{x,r}\otimes U^B_{y,r}\otimes I_{C_0})\sigma_r
  (U^{A*}_{x,r}\otimes U^{B*}_{y,r}\otimes I_{C_0})
 \right],
\]
which is exactly the referee's state in the original protocol.  This uses only that the two local
unitary operators act on disjoint registers, not that $\sigma_r$ is a product state.  Part (ii) is
symmetric.
\end{proof}

\subsection{Min-entropy sampling}

We use the technique essentially derived in \cite{Wul11} to obtain~\cref{lem:wullschleger}. Its proof, based on the approximate list-decoding method~\cite[Theorem 44]{IJK09}, is given in Appendix~\ref{app:wullschleger}.

\begin{Lem}[{\cite{Wul11}}]
\label{lem:wullschleger}
For any constant $0<\gamma\le1$,
there are constants
$\lambda_\gamma,c_\gamma>0$ and $k_\gamma\in\mathbb N$ such that, if
\[
 H_{\min}(V\mid Q)\ge\gamma N
\]
for $V\in\{0,1\}^N$ and quantum side information $Q$ and $k_\gamma\le k\le\lambda_\gamma N$ then
\begin{equation}
 p_{\mathrm{guess}}(V_T\mid TQ)\le2^{-c_\gamma k}
 \label{eq:min-entropy-sampling}
\end{equation}
where $T$ is a uniformly random $k$-subset of $[N]$, independent of $VQ$.
\end{Lem}
Taking $k=\lfloor\lambda_\gamma N\rfloor$ gives an exponential upper bound on
$p_{\mathrm{guess}}(V_T\mid TQ)$ for all sufficiently large $N$.

We next translate Lemma~\ref{lem:wullschleger} from sampled bits $V_T$ to
one-way communication for  $g^n(X^n,Y^n)$ in~\cref{lem:functional-sampling}.
\cref{lem:functional-sampling} considers a one-way communication problem where Alice (resp. Bob) is given $X^n \sim \Unif(\mathcal X^n)$ (resp. $Y^n \sim \Unif(\mathcal Y^n)$) uniformly at random, and Alice sends some quantum message in the register $U$ to Bob. 
Bob then perform a measurement $\mathsf M_{y^n}$ depending on $y^n$ to recover $g^n(x^n,y^n)$. The output from Bob is denoted by $\widehat{W}^n$.

\begin{Lem}
\label{lem:functional-sampling}
Assume a one-way communication strategy outputting $\widehat{W}^n$ above.
Let $\rev{\mathcal X},\rev{\mathcal Y},\mathcal W$ be finite sets, and let
$g:\rev{\mathcal X}\times\rev{\mathcal Y}\to\mathcal W$
have pairwise distinct rows; that is, the functions
$ g_x := g(x,\cdot), ~  (x\in\rev{\mathcal X})$ are pairwise distinct.

Then for every constant $h>0$, the condition
\begin{equation}
 H_{\min}(X^n\mid U)\ge hn
 \label{eq:functional-entropy-left}
\end{equation}
implies the average success probability
\[
 \AvSP_n :=\Prb\left(\widehat{W}^n \rev{=} g^n(X^n,Y^n)\right) \le2^{-\kappa n}
\]
for a constant $\kappa=\kappa(\rev{\mathcal X},\rev{\mathcal Y},\mathcal W,h)>0$ and all
sufficiently large $n$.
\end{Lem}

An alternative proof of \cref{lem:functional-sampling}, using the
pretty-good measurement directly, is given in Appendix~\ref{app:simplified-proof}.
\par

\begin{proof}
If $|\mathcal W|=1$, pairwise distinctness of the rows implies
$|\rev{\mathcal X}|=1$, yielding $H_{\min}(X^n\mid U)=0$.  Thus the hypothesis cannot
hold for $h>0$.  We may therefore assume $|\mathcal W|\ge2$.

To apply~\cref{lem:wullschleger}, we encode elements  $X^n \in \rev{\mathcal X}^n$ to $nL$ bit-strings $V_n$ as follows. 
Let $\rev{m_{\mathcal W}}:=\lceil\log|\mathcal W|\rceil$, fix an injective map
$e:\mathcal W\to\{0,1\}^{\rev{m_{\mathcal W}}}$, and put $L:=|\rev{\mathcal Y}|\rev{m_{\mathcal W}}$.  It is convenient to
index the $L$ bits in one row of the truth table by pairs
$(b,j)\in\rev{\mathcal Y}\times[\rev{m_{\mathcal W}}]$.  Define
\[
 C(x)_{b,j}:=e(g(x,b))_j
 \quad\text{and}\quad
 V_n:=\bigl(C(X_1),\ldots,C(X_n)\bigr)\in\{0,1\}^{nL}.
\]
The rows $g_x$ are pairwise distinct and $e$ is injective, so $C$ is injective.
The classical variables $X^n$ and $V_n(X^n)$ are therefore injective relabelings of
one another.  Consequently,
\begin{equation}
 H_{\min}(V_n(X^n)\mid U)=H_{\min}(X^n\mid U)\ge hn.
 \label{eq:truth-table-entropy}
\end{equation}

Since $V_n(X^n)$ has $nL$ bits, \cref{Fact_cq-state_up} and
\eqref{eq:truth-table-entropy} imply $h\le L$.
Thus, with $\gamma:=h/L$, we have
$0<\gamma\le1$.  
Take the constants $\lambda_\gamma,c_\gamma,k_\gamma$ from Lemma~\ref{lem:wullschleger}, and fix
\begin{equation}
 0<\eta\le\min\{1/4,L\lambda_\gamma,c_\gamma/4\}.
 \label{eq:functional-eta}
\end{equation}
Let $k:=\lfloor\eta n\rfloor$ and let $T$ be a uniform $k$-subset of the $nL$ bit positions,
independent of $V_n(X^n) U$.  Once
\begin{equation}
 n\ge \frac{k_\gamma+1}{\eta},
 \label{eq:functional-k-threshold}
\end{equation}
we have
\[
 k_\gamma\le k\le \eta n\le\lambda_\gamma nL.
\]
Lemma~\ref{lem:wullschleger} thus applies and gives
\begin{equation}
 p_{\mathrm{guess}}(V_T(X^n)\mid TU)
 \le2^{-c_\gamma k}.
 \label{eq:sample-upper}
\end{equation}

We next upper bound $\AvSP_n$ by $p_\mathrm{guess}$ using the one-way protocol that outputs $\widehat{W}^n$ with measurements $\{\mathsf M_{y^n}\}_{y^n\in\rev{\mathcal Y}^n}$. 
To this end, we construct a guessing protocol for $V_T$ based on the one-way protocol together with $T$. 

\paragraph{A new decoder for $V_t(X^n)$}
Identify a bit position of elements in $\Bset^{nL}$ with a triple $(i,b,j)\in[n]\times\rev{\mathcal Y}\times[\rev{m_{\mathcal W}}]$, and define the set of good
samples by
\begin{equation}
 \mathcal G_k
 :=\left\{t\subseteq[n]\times\rev{\mathcal Y}\times[\rev{m_{\mathcal W}}]:
 |t|=k\ \text{and}\
 \bigl|t\cap(\{i\}\times\rev{\mathcal Y}\times[\rev{m_{\mathcal W}}])\bigr|\le1
 \text{ for every }i\in[n]\right\}.
 \label{eq:good-samples}
\end{equation}
Thus $t\in\mathcal G_k$ means that no two sampled bits come from the same
coordinate of $X^n$.  For such a $t$, construct a decoder for $V_t(X^n)$
coordinate-wise as summarized in~\cref{tab:sampled-bit-decoder}.

\begin{table}[t]
\centering
\renewcommand{\arraystretch}{1.25}
\begin{tabular}{c|c|c}
Status of coordinate $i$
  & Choice of $Y_i$
  & Decoder output from $i$ \\
\hline
$\exists!\,(b,j)$ such that $(i,b,j)\in t$
  & $Y_i=b$
  & $e(\widehat O_i)_j$ \\
No $(b,j)$ satisfies $(i,b,j)\in t$
  & $Y_i\sim\Unif(\rev{\mathcal Y})$, independently
  & None
\end{tabular}
\caption{Coordinate-wise construction of the decoder for $V_t(X^n)$ when
$t\in\mathcal G_k$.}
\label{tab:sampled-bit-decoder}
\end{table}

After completing $Y^n$ according to the second column, apply
$\mathsf M_{Y^n}$ to $U$ and denote its outcome by $\widehat O^n$.  The decoder
collects the entries in the third column and outputs
\[
 \widehat V_t
 :=\bigl(e(\widehat O_i)_j:(i,b,j)\in t\bigr)
\]
Private randomness in the completion of $Y^n$ can be included in the decoder's
POVM.

\paragraph{Probability comparison.}
We first argue that the marginal distribution $Y^n$ is uniform conditioned on $T \in \mathcal G_k$.
For $t\in\mathcal G_k$, let
\[
 \mathcal Y(t)
 :=\{y^n\in\rev{\mathcal Y}^n:
       y_i=b\text{ whenever }(i,b,j)\in t\}.
\]
Thus, conditioned on $T=t$, the decoder chooses $Y^n$ uniformly from
$\mathcal Y(t)$.  If $T$ is uniform conditioned on $T\in\mathcal G_k$, it may
be generated by first choosing $k$ distinct coordinates and then choosing a
uniform pair $(b,j)\in\rev{\mathcal Y}\times[\rev{m_{\mathcal W}}]$ at each selected coordinate.  Hence
the prescribed value $Y_i=b$ is uniform at every selected coordinate, while
the decoder independently chooses $Y_i$ uniformly at every unselected
coordinate.  It follows that the marginal distribution of $Y^n$ is uniform on
$\rev{\mathcal Y}^n$ after averaging over $T\in_U\mathcal G_k$.

For a fixed $t\in\mathcal G_k$, the decoder applies the same measurement
$\mathsf M_{y^n}$ to the same state $U$ as the original one-way protocol, but
with $Y^n$ distributed uniformly on $\mathcal Y(t)$.  Therefore the conditional
law of its intermediate outcome $\widehat O^n$ agrees with that of the original
output $\widehat W^n$ under this restricted distribution of $Y^n$.  Together
with the preceding marginal-distribution observation, this gives
\begin{equation}
 \E_{t\in_U\mathcal G_k}
 \Prb\bigl[
   \widehat O^n=g^n(X^n,Y^n)\mid T=t
 \bigr]
 =\Prb\bigl[
   \widehat W^n=g^n(X^n,Y^n)
 \bigr]
 = \AvSP_n.
 \label{eq:lem_Succ-Prob_1}
\end{equation}

We now relate success of the original measurement to success of the new
decoder.  Fix $t\in\mathcal G_k$ and $(i,b,j)\in t$.  Since the decoder sets
$Y_i=b$, the event $\widehat O^n=g^n(X^n,Y^n)$ implies
\[
 e(\widehat O_i)_j
 =e(g(X_i,Y_i))_j
 =e(g(X_i,b))_j
 =C(X_i)_{b,j}
 =V_n(X^n)_{(i,b,j)}.
\]
The same implication holds for every triple in $t$, and hence
\begin{equation}
 \widehat O^n=g^n(X^n,Y^n)
 \quad\Longrightarrow\quad
 \widehat V_t=V_t(X^n).
 \label{eq:decoder-event-inclusion}
\end{equation}

Since $T$ is available to the guessing measurement, the optimal guessing
probability is at least the success probability of the constructed decoder on
the good samples.  Using \eqref{eq:decoder-event-inclusion} and then
\eqref{eq:lem_Succ-Prob_1}, we obtain
\begin{align}
 p_{\mathrm{guess}}(V_T(X^n)\mid TU)
 &\ge
 \sum_{t\in\mathcal G_k}\Prb[T=t]
 \Prb\bigl[\widehat V_t=V_t(X^n)\mid T=t\bigr]
 \nonumber\\
 &\ge
 \Prb[T\in\mathcal G_k] \times 
 \E_{t\in_U\mathcal G_k}
 \Prb\bigl[
   \widehat O^n=g^n(X^n,Y^n)\mid T=t
 \bigr]
 \nonumber\\
 &=\Prb[T\in\mathcal G_k] \times \AvSP_n.
 \label{eq:sample-lower}
\end{align}

It remains to bound the probability in \eqref{eq:sample-lower}.  A good sample is obtained by
choosing $k$ of the $n$ coordinate blocks and one of the $L$ positions in each chosen block.  Since
$T$ is uniform among all $k$-subsets of the $nL$ positions,
\begin{equation}
 \Prb[T\in\mathcal G_k]
 =\frac{\binom nk L^k}{\binom{nL}{k}}
 =\prod_{j=0}^{k-1}\frac{L(n-j)}{nL-j}
 =\prod_{j=0}^{k-1}\left(1-\frac{(L-1)j}{nL-j}\right).
 \label{eq:good-probability}
\end{equation}
By \eqref{eq:functional-eta}, $k\le n/4$.  Hence
\[
 0\le \frac{(L-1)j}{nL-j}\le\min\left\{\frac14,\frac jn\right\}
 \qquad(0\le j<k).
\]
Using $-\ln(1-u)\le2u$ for $0\le u\le1/4$ in \eqref{eq:good-probability}, we obtain
\begin{equation*}
 -\log\Prb[T\in\mathcal G_k]
 \le \frac{2}{\ln2}\sum_{j=0}^{k-1}\frac jn
 \le\frac{2k^2}{n}
\end{equation*}
meaning
\begin{equation}
 \frac{1}{\Prb[T\in\mathcal G_k]}
 \le 2^{2k^2/n}.
 \label{eq:good-bound}
\end{equation}

Combining \eqref{eq:sample-upper}, \eqref{eq:sample-lower}, and \eqref{eq:good-bound} yields
\[
 \AvSP_n \le \frac{p_{\mathrm{guess}}(V_T(X^n)\mid TU)}{\Prb[T \in \mathcal G_k]} \le 2^{-c_\gamma k+2k^2/n}.
\]
Since $k\ge\eta n-1$ and $k\le\eta n$, $\eta \le \min\{1/4, c_\gamma/4\}$ in  \eqref{eq:functional-eta} implies
\[
 -c_\gamma k+\frac{2k^2}{n}
 \le -c_\gamma (\eta n - 1) + \frac{2\eta^2 n^2}{n}
 \le -c_\gamma (\eta n - 1) + 2\eta n \frac{c_\gamma}{4}
 \le-\frac{c_\gamma\eta}{2}n+c_\gamma.
\]
For all sufficiently large $n$, this is at most $-\kappa n$ with
$\kappa:=c_\gamma\eta/4>0$, proving the lemma.
\end{proof}

\subsection{Communication leaves min-entropy}

We first treat a message without entanglement.

\begin{Lem}
\label{lem:sr-entropy}
Let $A$ be uniform on a set of size $R$, and let $M$ be a $d$-dimensional quantum encoding of $A$.
Then
\[
 p_{\mathrm{guess}}(A\mid M)\le\frac dR,
 \qquad
 H_{\min}(A\mid M)\ge\log R-\log d.
\]
\end{Lem}

\begin{proof}
For every POVM $\{E_a\}$ and every density matrix $\rho_a^M$, the inequality $\rho_a\le I_M$
gives
\[
 \frac1R\sum_a\operatorname{Tr}(E_a\rho_a)
 \le\frac1R\sum_a\operatorname{Tr}E_a
 =\frac dR.
\]
\end{proof}

The entangled case uses the following fixed-marginal bound.

\begin{Lem}
\label{lem:fixed-marginal}
Let $A$ be uniform on a set of size $R$.  For each $a$, let $\rho_a^{MS}$ be a state such that
$\rho_a^S=\sigma^S$ is independent of $a$, and let $d=\dim M$.  Then
\[
 p_{\mathrm{guess}}(A\mid MS)\le\frac{d^2}{R},
 \qquad
 H_{\min}(A\mid MS)\ge\log R-2\log d.
\]
\end{Lem}

\begin{proof}
Every state $\rho^{MS}$ with $\dim M=d$ satisfies
\begin{equation}
 \rho^{MS}\le d\,I_M\otimes\rho^S.
 \label{eq:operator-bound}
\end{equation}
Indeed, write $\rho^{MS}$ as a positive block matrix $(\rho_{ij})_{i,j=1}^d$.  For
$|v\rangle=\sum_i|i\rangle|v_i\rangle$, positivity and Cauchy--Schwarz give
\begin{align*}
 \langle v|\rho|v\rangle
 &\le\left(\sum_i\sqrt{\langle v_i|\rho_{ii}|v_i\rangle}\right)^2\\
 &\le d\sum_i\langle v_i|\rho_{ii}|v_i\rangle
 \le d\sum_i\langle v_i|\rho^S|v_i\rangle.
\end{align*}
This proves \eqref{eq:operator-bound}.  For a POVM $\{E_a\}$ on $MS$, apply the bound to every
$\rho_a$ and use their common marginal:
\begin{align*}
 \frac1R\sum_a\operatorname{Tr}(E_a\rho_a)
 &\le\frac dR\sum_a\operatorname{Tr}\bigl(E_a(I_M\otimes\sigma)\bigr)\\
 &=\frac dR\operatorname{Tr}(I_M\otimes\sigma)
 =\frac{d^2}{R}.
\end{align*}
\end{proof}

\begin{Prop}
\label{prop:one-way-quantum}
Let $f:\mathcal X\times\mathcal Y\to\mathcal W$ be a finite total function.  Let
$\rev{\mathcal X_0}\subseteq\mathcal X$ contain one equivalence class of each row of $f$, and let
$\mathcal Y_0\subseteq\mathcal Y$ be such that the restricted rows
\[
 f(x,\cdot)|_{\mathcal Y_0},\qquad x\in\rev{\mathcal X_0},
\]
are pairwise distinct.  Fix $\mathrm{res}\in\{\mathrm{SR},\EA\}$ and $\delta>0$.  There is
$\kappa=\kappa(f,\mathcal Y_0,\delta,\mathrm{res})>0$ such that every fixed-length quantum
one-way protocol for $f^n$ satisfying
\begin{equation}
 q:=\max_r\log\dim M_r
 \le n\left(\frac{\log |\Rf|}{\alpha_{\mathrm{res}}}-\delta\right)
 \label{eq:one-way-function-threshold}
\end{equation}
has success probability at most $2^{-\kappa n}$ under
\[
 \operatorname{Unif}(\rev{\mathcal X_0}^n)
 \otimes
 \operatorname{Unif}(\mathcal Y_0^n)
\]
for all sufficiently large $n$. 
\end{Prop}

\begin{proof}
Let $X^n$ and $Y^n$ be independent and uniform on $\rev{\mathcal X_0}^n$ and
$\mathcal Y_0^n$, respectively, and put
\[
 f_0:=f|_{\rev{\mathcal X_0}\times\mathcal Y_0}.
\]
By assumption, the rows of $f_0$ are pairwise distinct.

Fix an arbitrary shared-randomness value $R=r$.  Let $M_r$ be Alice's message register in this
instance and put $q_r:=\log\dim M_r\le q$.

In SR, let $\rho_{x^n,r}^{M_r}$ be Alice's message on input $x^n$.  By
Lemma~\ref{lem:sr-entropy},
\[
 H_{\min}(X^n\mid M_r)
 \ge n\log |\Rf|-q_r
 \ge\delta n.
\]
For every fixed $y^n$, Bob's decoder conditioned on $R=r$ is a measurement on $M_r$.  Applying
Lemma~\ref{lem:functional-sampling} to $f_0$, with $U=M_r$ and $h=\delta$, bounds the
conditional success probability in this value of shared randomness by $2^{-\kappa n}$.

In \EA, write the prior state in this value of shared randomness as $\sigma_r^{A_0B_0}$ and let $U_r=M_rB_0$ be the
state after Alice's encoding and transmission of $M_r$.  The marginal on $B_0$ is independent of
$x^n$, because Alice acts locally by a trace-preserving channel.  Lemma~\ref{lem:fixed-marginal}
therefore gives
\[
 H_{\min}(X^n\mid U_r)
 \ge n\log |\Rf|-2q_r
 \ge2\delta n.
\]
For every fixed $y^n$, Bob's decoder conditioned on $R=r$ is a measurement on $U_r$.  The same
application of Lemma~\ref{lem:functional-sampling} to $f_0$, now with $h=2\delta$, bounds the
conditional success probability by $2^{-\kappa n}$.

In either case, $\kappa$ depends only on $f$, $\mathcal Y_0$, $\delta$, and $\mathrm{res}$, and is
therefore independent of $r$, the initial shared state, and the protocol.  Averaging the conditional
success probabilities over $R$ proves the claim.
\end{proof}

\subsection{Completion of the proof}

\begin{proof}[Proof of Theorem~\ref{thm:general-relation-strong}]
If $\Dpar{g}=0$, assumption \eqref{eq:general-relation-strong-cost} is impossible, so assume $\Dpar{g}>0$.
By \cref{lem:fixed-register-reduction}, $\Pi_n$ has a fixed-register simulation $\widehat\Pi_n$
with the same output distribution and
\[
 \operatorname{QCost}_{\mathrm{wc}}(\widehat\Pi_n)
 \le
 n\left(\frac{\Dpar{g}}{\alpha_{\mathrm{res}}}-\Delta\right)+2.
\]
For all sufficiently large $n$, the right-hand side is at most
\[
 n\left(\frac{\Dpar{g}}{\alpha_{\mathrm{res}}}-\frac\Delta2\right).
\]
Apply \cref{lem:relation-to-function} to $\widehat\Pi_n$, and denote the resulting protocol for $g^n$ by $\widetilde\Pi_n$.  
It has success probability at least that of $\widehat\Pi_n$, encoders, shared resource, and
communication as $\widehat\Pi_n$.

Fix a public-randomness value $R=r$, and let $q_A(r)$ and $q_B(r)$ be the two fixed message
lengths in this value of shared-randomness.  They satisfy
\[
 q_A(r)+q_B(r)
 \le n\left(
  \frac{\log |\Rg|+\log|\Cg|}{\alpha_{\mathrm{res}}}-\frac\Delta2
 \right).
\]
Consequently, at least one of
\begin{align}
 q_A(r)&\le n\left(\frac{\log |\Rg|}{\alpha_{\mathrm{res}}}-\frac\Delta4\right),
 \label{eq:alice-short-quantum}\\
 q_B(r)&\le n\left(\frac{\log |\Cg|}{\alpha_{\mathrm{res}}}-\frac\Delta4\right)
 \label{eq:bob-short-quantum}
\end{align}
holds.  When $\rev{|\Rg|=1}$, the first inequality cannot hold, and the analogous statement applies
when $\rev{|\Cg|=1}$.

Suppose that \eqref{eq:alice-short-quantum} holds.  By
Lemma~\ref{lem:one-sender-simulation}, the Alice-only state is a one-way encoding from Alice to a
receiver holding $B_0C_0$: for each $y^n$, the receiver can reproduce the referee's output.  If the
SMP resource is SR, this is an SR one-way protocol; if it is BE or TE, it is an \EA\ one-way protocol
across the bipartition $A_0:(B_0C_0)$.  The set $\rev{\mathcal Y_g}$ of column equivalence classes separates
the rows indexed by $\mathcal X_g$.  Proposition~\ref{prop:one-way-quantum}, applied to $f=g$ with
$\mathcal Y_0=\mathcal Y_g$, therefore bounds the success probability of $\widetilde\Pi_n$
averaged under $\UnifXYng$ by $2^{-\kappa_A n}$.

If \eqref{eq:bob-short-quantum} holds, apply the same argument to the transposed function and the
Bob-only state to obtain $2^{-\kappa_B n}$.  These constants are independent of the
public-randomness value and of the shared state.  Therefore, with
\[
 \kappa:=\min\{\kappa_A,\kappa_B\},
\]
where an inapplicable term is omitted, every public-randomness instance has average success at most
$2^{-\kappa n}$.  Averaging over $R$ gives the same bound for $\widetilde\Pi_n$.  The pointwise
comparison \eqref{eq:relation-success-below-function} now proves
\eqref{eq:general-relation-strong-success}.

The classical assertion follows from \cref{lem:classical-as-sr-quantum}.
\end{proof}

\section{Exact asymptotic rates}
\label{sec:tight-characterization}

\begin{Thm}
\label{thm:matching-relation-rates}
Let $\Rel$ be a finite total relation, and suppose that $g$, $\phi$, and $s$ satisfy
Condition~\cref{cond:2}.  For every $0\le\varepsilon<1$, the following limits hold; in the
quantum formulas, $\mathrm{res}\in\{\mathrm{SR},\mathrm{BE},\mathrm{TE}\}$:
\begin{align}
 \lim_{n\to\infty}\frac{\ER^{\SMP}_\varepsilon(\Rel^n)}{n}
 &=(1-\varepsilon)\Dpar{g}, \label{eq:relation-classical-expected}\\
 \lim_{n\to\infty}\frac{\Qbar^{\mathrm{res}}_\varepsilon(\Rel^n)}{n}
 &=\frac{1-\varepsilon}{\alpha_{\mathrm{res}}}\Dpar{g},
 \label{eq:relation-quantum-expected}\\
 \lim_{n\to\infty}\frac{R^{\SMP,\mathrm{wc}}_\varepsilon(\Rel^n)}{n}
 &=\Dpar{g}, \label{eq:relation-classical-wc}\\
 \lim_{n\to\infty}\frac{Q^{\mathrm{res},\mathrm{wc}}_\varepsilon(\Rel^n)}{n}
 &=\frac{\Dpar{g}}{\alpha_{\mathrm{res}}}.
 \label{eq:relation-quantum-wc}
\end{align}
\end{Thm}
The proof of~\cref{thm:matching-relation-rates} consists of the upper bounds (\cref{subsec:Upper}), the lower bounds in worst-case (\cref{subsec:Lower_worst}) and expected communication (\cref{subsec:Lower_exp}). 
The upper bounds for both worst-case and expected case and the lower bounds for the worst-case relatively follow easily, while the lower bounds in expected communication requires some elementary lemmas.

\cref{thm:matching-relation-rates} applies to every relation of the form
\eqref{eq:function-defined-relation}.  This includes the singleton relation associated with every
total function as in~\cref{cor:function-rates} and the centralized winning relation of every finite linear game, including CHSH.

\begin{Cor}
\label{cor:function-rates}
Let $f:\mathcal X\times\mathcal Y\to\mathcal W$ be a finite total function.  For every
$0\le\varepsilon<1$ and every
$\mathrm{res}\in\{\mathrm{SR},\mathrm{BE},\mathrm{TE}\}$,
\begin{align}
 \lim_{n\to\infty}\frac{\ER^{\SMP}_\varepsilon(f^n)}{n}
 &=(1-\varepsilon)\Dpar{f}, \label{eq:function-classical-expected}\\
 \lim_{n\to\infty}\frac{\Qbar^{\mathrm{res}}_\varepsilon(f^n)}{n}
 &=\frac{1-\varepsilon}{\alpha_{\mathrm{res}}}\Dpar{f},
 \label{eq:function-quantum-expected}\\
 \lim_{n\to\infty}\frac{R^{\SMP,\mathrm{wc}}_\varepsilon(f^n)}n
 &=\Dpar{f}, \label{eq:function-classical-wc}\\
 \lim_{n\to\infty}\frac{Q^{\mathrm{res},\mathrm{wc}}_\varepsilon(f^n)}n
 &=\frac{\Dpar{f}}{\alpha_{\mathrm{res}}}.
 \label{eq:q-wc-function}
\end{align}
\end{Cor}

\subsection{Upper bounds}\label{subsec:Upper}

Alice sends the sequence of rows
\[
 (g_{x_1},\ldots,g_{x_n})\in\operatorname{Row}(g)^n,
\]
and Bob sends the sequence of columns
\[
 (g^{y_1},\ldots,g^{y_n})\in\operatorname{Col}(g)^n.
\]
The referee outputs $s(g_{x_i},g^{y_i})$ in coordinate $i$.  The left-hand side of
\eqref{eq:relation-structure} makes this a zero-error protocol.  Fixed-length encodings use at most
\[
 \lceil n\log |\Rg|\rceil+\lceil n\log |\Cg|\rceil
\]
classical bits.  Orthogonal encoding gives the same SR quantum rate.  In BE and TE, each sender
uses EPR pairs shared with the referee and superdense coding, giving total quantum communication
$n\Dpar{g}/2+O(1)$. This proves the upper bounds for the worst-case statements, in
\eqref{eq:relation-classical-wc} and \eqref{eq:relation-quantum-wc}.

For expected communication, consider a public Bernoulli variable with which a protocol runs the zero-error protocol with
probability $1-\varepsilon$ and otherwise sends no messages and produces an arbitrary output. This clearly has the expected cost as desired,
and therefore proves the upper bounds in \eqref{eq:relation-classical-expected} and
\eqref{eq:relation-quantum-expected}.  

\subsection{Lower bounds for worst-case communication}\label{subsec:Lower_worst}

\begin{proof}[Worst-case lower bounds in Theorem~\ref{thm:matching-relation-rates}]
Theorem~\ref{thm:general-relation-strong} gives the quantum lower bounds: a protocol with pointwise
success at least $1-\varepsilon$ has average success at least $1-\varepsilon$ under $\UnifXYng$,
contradicting \eqref{eq:general-relation-strong-success} at every rate below the stated threshold.

The classical lower bound follows from the SR quantum lower bound and
\cref{lem:classical-as-sr-quantum}.  Together with the upper bounds above, this proves
\eqref{eq:relation-classical-wc} and \eqref{eq:relation-quantum-wc}.
\end{proof}

\subsection{Lower bounds for expected communication}\label{subsec:Lower_exp}

The following elementary lemma converts the probability mass of one exact pair of
quantum message lengths into the success probability of a fixed-length protocol.

\begin{Lem}
\label{lem:exact-length-filtering}
Let $P$ be a finite total relation, let $\Pi$ be a variable-length quantum SMP protocol for $P^n$
with resource $\mathrm{res}\in\{\mathrm{SR},\mathrm{BE},\mathrm{TE}\}$, and let $\mu$ be an input distribution.  For every pair of
nonnegative integers $(a,b)$, there is a quantum SMP protocol $\Pi_{a,b}$ with the same shared
resource as $\Pi$ such that
\[
 \operatorname{QCost}_{\mathrm{wc}}(\Pi_{a,b})\le a+b+2
\]
and
\begin{align}
 &\Prb\bigl[\Pi_{a,b}(X^n,Y^n)\in P^n(X^n,Y^n)\bigr]
 \notag\\
 &\qquad\ge
 \Prb\bigl[\Pi(X^n,Y^n)\in P^n(X^n,Y^n),
             \ q_A(X^n,R)=a,\ q_B(Y^n,R)=b\bigr],
 \label{eq:filtered-success}
\end{align}
where the probability on the right is over $(X^n,Y^n)\sim\mu$, the shared randomness, and the
protocol output.
\end{Lem}

\begin{proof}
	Fix $R=r$ and apply~the same fixed-length construction~(\cref{lem:fixed-register-reduction}) for an instance $(a,b)$.  If $q_A(x,r)=a$, Alice sends the resulting $a$-qubit message in a fixed $2^a$-dimensional subspace of an fixed $(a+1)$-qubit register.  Otherwise she sends a fixed state $\lvert\bot\rangle$ in an
orthogonal subspace.  Bob applies the analogous construction with an $(b+1)$-qubit register.

The referee first checks whether either message is $\lvert\bot\rangle$.  If not, it applies the
original measurement for lengths $(a,b)$; otherwise it outputs an arbitrary value.  When the two
original lengths equal $(a,b)$, the output distribution is unchanged implying that 
\begin{align}
	&``\Pi(X^n,Y^n)\in P^n(X^n,Y^n),\ q_A(X^n,R)=a,\ q_B(Y^n,R)=b" \notag\\
	&\quad \Longrightarrow
``\Pi_{a,b}(X^n,Y^n)\in P^n(X^n,Y^n)"
\end{align}

This proves
\eqref{eq:filtered-success}.  The shared state is unchanged, and the worst-case communication is at
most $a+b+2$.
\end{proof}

\begin{proof}[Expected lower bounds in Theorem~\ref{thm:matching-relation-rates}]
The claims are immediate when $\Dpar{g}=0$, so assume $\Dpar{g}>0$.  Fix a resource model
$\mathrm{res}\in\{\mathrm{SR},\mathrm{BE},\mathrm{TE}\}$ and $0<\Delta<\Dpar{g}/\alpha_{\mathrm{res}}$.  Let $\Pi_n$ be any quantum protocol for
$\Rel^n$ with that resource and with pointwise global error at most $\varepsilon$.  Write
\[
 C_n:=\operatorname{QCost}_{\rev{\mathrm{exp}}}(\Pi_n),
 \qquad
 T_n:=\left\lfloor n\left(\frac{\Dpar{g}}{\alpha_{\mathrm{res}}}-\Delta\right)\right\rfloor-2.
\]
Sample $(X^n,Y^n)\sim\UnifXYng$ independently of $R$.  Let $S$ be the event of
relation success, and put
\[
 L_A:=q_A(X^n,R),\qquad L_B:=q_B(Y^n,R),\qquad L:=L_A+L_B.
\]
Pointwise correctness and the definition of worst-input expected communication give
\begin{equation}
 \Prb[S]\ge1-\varepsilon,
 \qquad
 \E[L]\le C_n.
 \label{eq:expected-success-cost}
\end{equation}

For every $a,b\ge0$ with $a+b\le T_n$, \cref{lem:exact-length-filtering} gives a protocol with
resource $\mathrm{res}$ and worst-case cost at most
\[
 a+b+2\le n\left(\frac{\Dpar{g}}{\alpha_{\mathrm{res}}}-\Delta\right).
\]
Theorem~\ref{thm:general-relation-strong} therefore gives a constant
$\kappa=\kappa(g,\Delta,\mathrm{res})>0$, independent of $(a,b)$, such that
\begin{equation}
 \Prb[S,\ L_A=a,\ L_B=b]\le2^{-\kappa n}
 \label{eq:one-length-pair}
\end{equation}
for all sufficiently large $n$.  There are
\[
 N_n:=\bigl|\{(a,b)\in\mathbb N_0^2:a+b\le T_n\}\bigr|
 =\frac{(T_n+1)(T_n+2)}2=O(n^2)
\]
such pairs.  Summing \eqref{eq:one-length-pair} and using Markov's inequality with
\eqref{eq:expected-success-cost}, we obtain
\begin{align*}
 1-\varepsilon
 &\le \Prb[S]\\
 &\le \Prb[L>T_n]+\Prb[S,\ L\le T_n]\\
 &\le \frac{C_n}{T_n}+N_n2^{-\kappa n}.
\end{align*}
Consequently,
\[
 C_n\ge T_n\bigl(1-\varepsilon-N_n2^{-\kappa n}\bigr).
\]
Since $N_n2^{-\kappa n}=o(1)$ and
$T_n/n\to \Dpar{g}/\alpha_{\mathrm{res}}-\Delta$,
\[
 \liminf_{n\to\infty}
 \frac{\Qbar^{\mathrm{res}}_\varepsilon(\Rel^n)}n
 \ge(1-\varepsilon)
 \left(\frac{\Dpar{g}}{\alpha_{\mathrm{res}}}-\Delta\right).
\]
Letting $\Delta\downarrow0$ proves \eqref{eq:relation-quantum-expected}.  Taking
$\mathrm{res}=\mathrm{SR}$ and using \cref{lem:classical-as-sr-quantum} gives
\[
 \liminf_{n\to\infty}
 \frac{\ER^{\SMP}_\varepsilon(\Rel^n)}n
 \ge(1-\varepsilon)\Dpar{g},
\]
which proves \eqref{eq:relation-classical-expected}.
\end{proof}

\section{One-way consequences}
\label{sec:one-way}
The one-way communication statements given in~\cref{cor:one-way-relation-strong,cor:one-way-relation-rates} are proven in simpler but essentially the same manner as the SMP case. We give a proof of them for completeness.
\begin{Cor}
\label{cor:one-way-relation-strong}
Let $\Rel$ be a finite total relation, and suppose that $g$ and $\phi$ satisfy Condition~\cref{cond:1}.
Fix
$\mathrm{res}\in\{\mathrm{SR},\mathrm EA\}$ and $\Delta>0$.  There is
$\kappa=\kappa(g,\Delta,\mathrm{res})>0$ such that every quantum one-way protocol $\Pi_n$ with
\begin{equation}
 \operatorname{QCost}^{\to}_{\mathrm{wc}}(\Pi_n)
 \le
 n\left(\frac{D_{\to}(g)}{\alpha_{\mathrm{res}}}-\Delta\right)
 \label{eq:one-way-relation-strong-cost}
\end{equation}
has, for all sufficiently large $n$,
\begin{equation}
	\E_{(X^n,Y^n)\sim\Unif(\mathcal{X}_g^n \times \mathcal Y^n)}
 \Prb\bigl[
   \Pi_n(X^n,Y^n)\in\Rel^n(X^n,Y^n)
   \mid X^n,Y^n
 \bigr]
 \le 2^{-\kappa n}.
 \label{eq:one-way-relation-strong-success}
\end{equation}
In particular, the conclusion holds for every public-coin classical one-way protocol of worst-case
cost at most $n(D_{\to}(g)-\Delta)$.
\end{Cor}

\begin{Cor}
\label{cor:one-way-relation-rates}
Let $\Rel$, $g$, $\phi$, and $s$ satisfy Condition~\cref{cond:2}.  For every
$0\le\varepsilon<1$,
\begin{align}
 \lim_{n\to\infty}
 \frac{\EC^{\to}_\varepsilon(\Rel^n)}n
 &=(1-\varepsilon)\log |\Rg|,
 \label{eq:one-way-classical-expected}\\
 \lim_{n\to\infty}
 \frac{\Qbar^{\to,\mathrm{SR}}_\varepsilon(\Rel^n)}n
 &=(1-\varepsilon)\log |\Rg|,
 \label{eq:one-way-quantum-sr-expected}\\
 \lim_{n\to\infty}
 \frac{\Qbar^{\to,\mathrm EA}_\varepsilon(\Rel^n)}n
 &=\frac{1-\varepsilon}{2}\log |\Rg|.
 \label{eq:one-way-quantum-e-expected}
\end{align}
For worst-case communication,
\begin{align}
 \lim_{n\to\infty}
 \frac{C^{\to,\mathrm{wc}}_\varepsilon(\Rel^n)}n
 &=\log |\Rg|,
 \label{eq:one-way-classical-wc}\\
 \lim_{n\to\infty}
 \frac{Q^{\to,\mathrm{SR},\mathrm{wc}}_\varepsilon(\Rel^n)}n
 &=\log |\Rg|,
 \label{eq:one-way-quantum-sr-wc}\\
 \lim_{n\to\infty}
 \frac{Q^{\to,\mathrm EA,\mathrm{wc}}_\varepsilon(\Rel^n)}n
 &=\frac12\log |\Rg|.
 \label{eq:one-way-quantum-e-wc}
\end{align}
\end{Cor}

The one-way result is obtained directly from Proposition~\ref{prop:one-way-quantum}.  In
particular, Bob does not send $y^n$ to an SMP referee, and the column term $\log |\Cg|$ is never
charged.

\begin{proof}[Proof of Corollary~\ref{cor:one-way-relation-strong}]
By \cref{lem:fixed-register-reduction}, $\Pi_n$ has a fixed-length simulation $\widehat\Pi_n$
with the same output distribution and
\[
 \operatorname{QCost}^{\to}_{\mathrm{wc}}(\widehat\Pi_n)
 \le
 n\left(\frac{D_{\to}(g)}{\alpha_{\mathrm{res}}}-\Delta\right)+1.
\]
For all sufficiently large $n$, this is at most
\[
 n\left(\frac{D_{\to}(g)}{\alpha_{\mathrm{res}}}-\frac\Delta2\right).
\]
By the observation~\eqref{eq:relation-structure_0} and postprocessing Bob's output coordinatewise by
$\phi$, we obtain a fixed-length protocol $\widetilde\Pi_n$ for $g^n$ with the same shared resource, communication, and success probability at least that of the original protocol.  
Hence Proposition~\ref{prop:one-way-quantum}, applied with $\mathcal A=\mathcal X_g$,
$\mathcal Y_0=\mathcal Y$, and gap $\Delta/2$, gives
\eqref{eq:one-way-relation-strong-success}.

The classical claim follows from \cref{lem:classical-as-sr-quantum}.
\end{proof}

We use the following one-message version of \cref{lem:exact-length-filtering}.

\begin{Lem}
\label{lem:one-way-exact-length-filtering}
Let $\Pi$ be a variable-length quantum one-way protocol for a relation $P^n$ with resource
$\mathrm{res}$, and let $\mu$ be an input distribution.  For every nonnegative integer $a$, there
is a quantum one-way protocol $\Pi_a$ with the same shared resource and worst-case cost at most
$a+1$ such that
\begin{align}
 &\Prb\bigl[\Pi_a(X^n,Y^n)\in P^n(X^n,Y^n)\bigr]
 \notag\\
 &\qquad\ge
 \Prb\bigl[\Pi(X^n,Y^n)\in P^n(X^n,Y^n),\ q(X^n,R)=a\bigr].
 \label{eq:one-way-filtered-success}
\end{align}
\end{Lem}

\begin{proof}
Fix $R=r$.  If $q(x,r)=a$, Alice applies her original encoding unitary and uses a local unitary
with ancillary registers to place the resulting message in a fixed $2^a$-dimensional subspace of
an $(a+1)$-qubit register.  Otherwise she sends a fixed orthogonal state
$\lvert\bot\rangle$.  Bob applies the original decoder unless he receives
$\lvert\bot\rangle$.  On the event $q(X^n,R)=a$, the output distribution is unchanged, which proves
\eqref{eq:one-way-filtered-success}.
\end{proof}

\begin{proof}[Proof of Corollary~\ref{cor:one-way-relation-rates}]
Write $D_{\to}=\log \Rg$.  For the upper bounds, Alice transmits the sequence of rows
\[
 (g_{x_1},\ldots,g_{x_n}).
\]
Bob then outputs
$s(g_{x_i},g^{y_i})$ in coordinate $i$.  This is a zero-error classical protocol by
\eqref{eq:relation-structure}, of cost
$nD_{\to}+O(1)$.  Orthogonal encoding gives the same SR quantum rate, while superdense coding with
prior EPR pairs gives E rate $D_{\to}/2$.  For expected communication, a public Bernoulli variable
runs this zero-error protocol with probability $1-\varepsilon$ and sends no message otherwise.
This proves all upper bounds.

We next prove the expected quantum lower bounds.  The claims are immediate if $D_{\to}=0$, so fix
$\mathrm{res}\in\{\mathrm{SR},\mathrm EA\}$ and
$0<\Delta<D_{\to}/\alpha_{\mathrm{res}}$.  Let $\Pi_n$ be a quantum protocol with this resource and
pointwise global error at most $\varepsilon$.  
Under $\Unif(\mathcal X^n_g \times \mathcal Y^n)$, let $S$ be the event of relation success, let $L:=q(X^n,R)$, and put
\[
 C_n:=\operatorname{QCost}^{\to}_{\mathrm{exp}}(\Pi_n),
 \qquad
 T_n:=\left\lfloor
 n\left(\frac{D_{\to}}{\alpha_{\mathrm{res}}}-\Delta\right)
 \right\rfloor-1.
\]
Then $\Prb[S]\ge1-\varepsilon$ and $\E[L]\le C_n$.  For every $0\le a\le T_n$,
Lemma~\ref{lem:one-way-exact-length-filtering} and
Corollary~\ref{cor:one-way-relation-strong} give a constant $\kappa>0$, independent of $a$, such
that
\[
 \Prb[S,\ L=a]\le2^{-\kappa n}
\]
for all sufficiently large $n$.  Hence Markov's inequality gives
\begin{align*}
 1-\varepsilon
	\le\Prb[S]
 &\le\Prb[L>T_n]+\sum_{a=0}^{T_n}\Prb[S,\ L=a]\\
 &\le\frac{C_n}{T_n}+(T_n+1)2^{-\kappa n}.
\end{align*}
It follows that
\[
 \liminf_{n\to\infty}
 \frac{\Qbar^{\to,\mathrm{res}}_\varepsilon(\Rel^n)}n
 \ge(1-\varepsilon)
 \left(\frac{D_{\to}}{\alpha_{\mathrm{res}}}-\Delta\right).
\]
Letting $\Delta\downarrow0$ proves \eqref{eq:one-way-quantum-sr-expected} and
\eqref{eq:one-way-quantum-e-expected}.  The SR case and
\cref{lem:classical-as-sr-quantum} also give the lower bound in
\eqref{eq:one-way-classical-expected}.

Finally, a pointwise-success protocol has average success at least $1-\varepsilon$ under
$\Unif(\mathcal X^n_f \times \mathcal Y^n)$.  Corollary~\ref{cor:one-way-relation-strong} rules out every worst-case rate
below $D_{\to}/\alpha_{\mathrm{res}}$.  The classical lower bound follows from the SR case and
\cref{lem:classical-as-sr-quantum}.  Together with the zero-error upper bounds, this proves
\eqref{eq:one-way-classical-wc}--\eqref{eq:one-way-quantum-e-wc}.
\end{proof}

\section*{Acknowledgement}
The author would like to thank Ashwin Nayak and Dave Touchette for their kindness and support; Atsuya Hasegawa and Masayuki Miyamoto for their monthly conversations.

The author acknowledges the support of the Natural Sciences and Engineering Research Council of Canada (NSERC) [funding reference no. ALLRP-578455-2022]. This research was supported in part by the Institute for Quantum Computing.

\bibliography{/Users/daikisuruga/Dropbox/Citations/mathematics_D, /Users/daikisuruga/Dropbox/Citations/Info_D, /Users/daikisuruga/Dropbox/Citations/computational_D, /Users/daikisuruga/Dropbox/Citations/query_D, /Users/daikisuruga/Dropbox/Citations/quant_info_D, /Users/daikisuruga/Dropbox/Citations/comm_comp_D, /Users/daikisuruga//Dropbox/Citations/books_D}

\begin{thebibliography}{BCWdW01}

\bibitem[AG23]{AG23}
Srinivasan Arunachalam and Uma Girish.
\newblock Trade-offs between entanglement and communication.
\newblock In {\em 38th Computational Complexity Conference}, CCC '23, 2023.

\bibitem[AHMS20]{AHMS20}
Benny Applebaum, Thomas Holenstein, Manoj Mishra, and Ofer Shayevitz.
\newblock The communication complexity of private simultaneous messages, revisited.
\newblock {\em Journal of Cryptology}, 33(3):917--953, 2020.

\bibitem[Ari73]{Ari73}
S.~Arimoto.
\newblock On the converse to the coding theorem for discrete memoryless channels (corresp.).
\newblock {\em IEEE Transactions on Information Theory}, 19(3):357--359, 1973.

\bibitem[AT16]{AT16}
Juan~Miguel Arrazola and Dave Touchette.
\newblock Quantum advantage on information leakage for equality.
\newblock {\em arXiv preprint arXiv:1607.07516}, 2016.

\bibitem[BCMdW10]{BCMdW10}
Harry Buhrman, Richard Cleve, Serge Massar, and Ronald de~Wolf.
\newblock Nonlocality and communication complexity.
\newblock {\em Reviews of modern physics}, 82:665--698, 2010.

\bibitem[BCWdW01]{BCWdW01}
Harry Buhrman, Richard Cleve, John Watrous, and Ronald de~Wolf.
\newblock Quantum fingerprinting.
\newblock {\em Physical Review Letters}, 87:167902, 2001.

\bibitem[BK97]{BK97}
L.~Babai and P.G. Kimmel.
\newblock Randomized simultaneous messages: solution of a problem of {Yao} in communication complexity.
\newblock In {\em 12th computational complexity conference}, pages 239--246, 1997.

\bibitem[BR11]{BR11}
Mark Braverman and Anup Rao.
\newblock Information equals amortized communication.
\newblock In {\em 52nd annual Symposium on Foundations of Computer Science}, pages 748--757, 2011.

\bibitem[Bra15]{Bra15}
Mark Braverman.
\newblock Interactive information complexity.
\newblock {\em SIAM Journal on Computing}, 44(6):1698--1739, 2015.

\bibitem[BW15]{BW15}
Mark Braverman and Omri Weinstein.
\newblock An interactive information odometer and applications.
\newblock In {\em 47th annual symposium on Theory of Computing}, pages 341--350, 2015.

\bibitem[CSWY01]{CSWY01}
Amit Chakrabarti, Yaoyun Shi, Anthony Wirth, and Andrew Yao.
\newblock Informational complexity and the direct sum problem for simultaneous message complexity.
\newblock In {\em 42nd annual Symposium on Foundations of Computer Science}, pages 270--278, 2001.

\bibitem[Gav08]{Gav08}
Dmitry Gavinsky.
\newblock On the role of shared entanglement.
\newblock {\em Quantum Information \& Computation}, 8(1):82--95, 2008.

\bibitem[GGJL25]{GGJL25}
Mika G{\"o}{\"o}s, Tom Gur, Siddhartha Jain, and Jiawei Li.
\newblock Quantum communication advantage in {TFNP}.
\newblock In {\em 57th annual Symposium on Theory of Computing}, pages 1465--1475, 2025.

\bibitem[GKDW06]{GKdW06}
Dmitry Gavinsky, Julia Kempe, and Ronald De~Wolf.
\newblock Strengths and weaknesses of quantum fingerprinting.
\newblock In {\em 21st computational complexity conference}, pages 8--pp. IEEE, 2006.

\bibitem[GKRdW06]{GKRdW06}
Dmitry Gavinsky, Julia Kempe, Oded Regev, and Ronald de~Wolf.
\newblock Bounded-error quantum state identification and exponential separations in communication complexity.
\newblock In {\em 38th annual Symposium on Theory of Computing}, pages 594--603, 2006.

\bibitem[GRT22]{GRT22}
Uma Girish, Ran Raz, and Avishay Tal.
\newblock Quantum versus randomized communication complexity, with efficient players.
\newblock {\em computational complexity}, 31(2):17, 2022.

\bibitem[IJK09]{IJK09}
Russell Impagliazzo, Ragesh Jaiswal, and Valentine Kabanets.
\newblock Approximate list-decoding of direct product codes and uniform hardness amplification.
\newblock {\em SIAM Journal on Computing}, 39(2):564--605, 2009.

\bibitem[Jai15]{Jai15}
Rahul Jain.
\newblock New strong direct product results in communication complexity.
\newblock {\em Journal of the ACM}, 62(3):1--27, 2015.

\bibitem[JK09]{JK09}
Rahul Jain and Hartmut Klauck.
\newblock New results in the simultaneous message passing model via information theoretic techniques.
\newblock In {\em 24th Computational Complexity Conference}, pages 369--378, 2009.

\bibitem[JK21]{JK21}
Rahul Jain and Srijita Kundu.
\newblock {A Direct Product Theorem for One-Way Quantum Communication}.
\newblock In {\em 36th Computational Complexity Conference}, volume 200, pages 27:1--27:28, 2021.

\bibitem[JK22]{JK22}
Rahul Jain and Srijita Kundu.
\newblock A direct product theorem for quantum communication complexity with applications to device-independent {QKD}.
\newblock In {\em 62nd annual Symposium on Foundations of Computer Science}, pages 1285--1295, 2022.

\bibitem[JRS03]{JRS03d}
Rahul Jain, Jaikumar Radhakrishnan, and Pranab Sen.
\newblock A direct sum theorem in communication complexity via message compression.
\newblock In {\em 30th International Colloquium on Automata, Languages and Programming}, pages 300--315, 2003.

\bibitem[Kla10]{Kla10}
Hartmut Klauck.
\newblock A strong direct product theorem for disjointness.
\newblock In {\em 42nd annual symposium on Theory of computing}, pages 77--86, 2010.

\bibitem[KN21]{KN21}
Akinori Kawachi and Harumichi Nishimura.
\newblock {Communication Complexity of Private Simultaneous Quantum Messages Protocols}.
\newblock In {\em 2nd Conference on Information-Theoretic Cryptography (ITC 2021)}, volume 199, pages 20:1--20:19, 2021.

\bibitem[KRS09]{KRS09}
Robert {K{\" o}nig}, Renato Renner, and Christian Schaffner.
\newblock The operational meaning of min-and max-entropy.
\newblock {\em IEEE Transactions on Information theory}, 55(9):4337--4347, 2009.

\bibitem[KT08]{KT08}
Robert~T. Konig and Barbara~M. Terhal.
\newblock The bounded-storage model in the presence of a quantum adversary.
\newblock {\em IEEE Transactions on Information Theory}, 54(2):749--762, 2008.

\bibitem[LWD16]{LWD16}
Felix Leditzky, Mark~M Wilde, and Nilanjana Datta.
\newblock Strong converse theorems using {R{\'e}nyi} entropies.
\newblock {\em Journal of Mathematical Physics}, 57(8), 2016.

\bibitem[Mer90]{Mer90}
N.~David Mermin.
\newblock Simple unified form for the major no-hidden-variables theorems.
\newblock {\em Physical Review Letters}, 65:3373--3376, Dec 1990.

\bibitem[NS96]{NS96}
Ilan Newman and Mario Szegedy.
\newblock Public vs. private coin flips in one round communication games.
\newblock In {\em 28th annual symposium on Theory of computing}, pages 561--570, 1996.

\bibitem[ON99]{ON99}
Tomohiro Ogawa and Hiroshi Nagaoka.
\newblock Strong converse to the quantum channel coding theorem.
\newblock {\em IEEE Transactions on Information Theory}, 45(7):2486--2489, 1999.

\bibitem[Per90]{Per90}
Asher Peres.
\newblock Incompatible results of quantum measurements.
\newblock {\em Physics Letters A}, 151(3):107--108, 1990.

\bibitem[RAM16]{RAM16}
Ravishankar Ramanathan, Remigiusz Augusiak, and Gl\'aucia Murta.
\newblock Generalized {XOR} games with $d$ outcomes and the task of nonlocal computation.
\newblock {\em Phys. Rev. A}, 93:022333, 2016.

\bibitem[She11]{She11SDP}
Alexander~A Sherstov.
\newblock Strong direct product theorems for quantum communication and query complexity.
\newblock In {\em 43rd annual symposium on Theory of computing}, pages 41--50, 2011.

\bibitem[Sur26]{Sur26}
Daiki Suruga.
\newblock Zero-error information equals amortized communication complexity.
\newblock {\em arXiv preprint arXiv:2608.04141}, 2026.
\newblock (To appear in {FOCS} 2026).

\bibitem[Tom15]{Tom15}
Marco Tomamichel.
\newblock {\em Quantum information processing with finite resources: mathematical foundations}.
\newblock Springer, 2015.

\bibitem[Tou15]{Tou15}
Dave Touchette.
\newblock Quantum information complexity.
\newblock In {\em 47th Annual Symposium on Theory of Computing}, pages 317--326, 2015.

\bibitem[TW19]{TW19}
Himanshu Tyagi and Shun Watanabe.
\newblock Strong converse using change of measure arguments.
\newblock {\em IEEE Transactions on Information Theory}, 66(2):689--703, 2019.

\bibitem[Wat18]{Wat18}
John Watrous.
\newblock {\em The theory of quantum information}, volume~1.
\newblock Cambridge university press, 2018.

\bibitem[Wol57]{Wol57}
Jacob Wolfowitz.
\newblock The coding of messages subject to chance errors.
\newblock {\em Illinois Journal of Mathematics}, 1(4):591--606, 1957.

\bibitem[Wul11]{Wul11}
J{\"u}rg Wullschleger.
\newblock Bitwise quantum min-entropy sampling and new lower bounds for random access codes.
\newblock In {\em Conference on Quantum Computation, Communication, and Cryptography}, pages 164--173, 2011.

\bibitem[Yao79]{Yao79}
Andrew Chi-Chih Yao.
\newblock Some complexity questions related to distributive computing~(preliminary report).
\newblock In {\em 11th annual Symposium on Theory of Computing}, pages 209--213, 1979.

\end{thebibliography}
\bibliographystyle{alpha}
\appendix
\section{Applications}
\label{sec:applications}

\subsection{Finding a common element or certifying disjointness}

Let $\mathcal X=\mathcal Y=2^{[m]}$ and
\begin{equation}
 \Rel_{\cap}(x,y)
 :=
 \begin{cases}
  \{\mathsf{disj}\},&x\cap y=\varnothing,\\
  x\cap y,&x\cap y\ne\varnothing.
 \end{cases}
 \label{eq:intersection-search}
\end{equation}
Define $\phi(\mathsf{disj})=0$ and $\phi(i)=1$ for every $i\in[m]$.
Then \cref{cond:1} holds for
\[
 g_{\cap}(x,y):=\mathbf 1[x\cap y\ne\varnothing].
\]
Distinct subsets give distinct rows and columns of $g_{\cap}$, so
\[
 \rev{|\operatorname{Row}(g_{\cap})|}
 =\rev{|\operatorname{Col}(g_{\cap})|}=2^m,
 \qquad \Dpar{g_{\cap}}=2m.
\]
Consequently, a row of $g_{\cap}$ uniquely determines $x$, and a column uniquely determines $y$.
Define the selector to output $\mathsf{disj}$ when $x\cap y=\varnothing$ and the least element of
$x\cap y$ otherwise.  This selector depends only on the row and column, so
\eqref{eq:relation-structure} holds.  Theorem~\ref{thm:matching-relation-rates} gives
\begin{align}
 \lim_{n\to\infty}
 \frac{\ER^{\SMP}_\varepsilon(\Rel_{\cap}^{\,n})}{n}
 &=2(1-\varepsilon)m,\\
 \lim_{n\to\infty}
 \frac{\Qbar^{\mathrm{res}}_\varepsilon(\Rel_{\cap}^{\,n})}{n}
 &=\frac{2(1-\varepsilon)m}{\alpha_{\mathrm{res}}}.
\end{align}
The worst-case rates are $2m$ and $2m/\alpha_{\mathrm{res}}$, respectively.

\subsection{Winning relations of linear games}

Let $G$ be a finite abelian group, let $h:\mathcal X\times\mathcal Y\to G$, and define
\begin{equation}
 \Rel_h(x,y):=\{(a,b)\in G^2:a+b=h(x,y)\}.
 \label{eq:linear-game-relation}
\end{equation}
This is of the form \eqref{eq:function-defined-relation}, with $g=h$ and
$\phi(a,b)=a+b$.  One may choose the canonical valid output $(0,h(x,y))$; the value $h(x,y)$ is
determined by the pair $(h_x,h^y)$.  Theorems~\ref{thm:general-relation-strong} and
\ref{thm:matching-relation-rates} therefore apply with
$\Dpar{h}=\rev{\log|\operatorname{Row}(h)|+\log|\operatorname{Col}(h)|}$.

In particular, the CHSH winning relation is obtained from $G=\mathbb Z_2$ and
\[
 h_{\mathrm{CHSH}}(x,y)=x\wedge y,
 \qquad
 \phi(a,b)=a\oplus b.
\]
It has $\rev{|\operatorname{Row}(h_{\mathrm{CHSH}})|}
=\rev{|\operatorname{Col}(h_{\mathrm{CHSH}})|}=2$ and hence
$\Dpar{h_{\mathrm{CHSH}}}=2$.  Here the SMP referee produces the pair $(a,b)$.  This centralized-output
communication problem should be distinguished from the usual nonlocal game, in which the two
separated players produce their own answers; see Ref.~\cite{RAM16}.

\subsection{The magic-square winning relation}

The preceding argument does not extend to all nonlocal games.  In the magic-square game~\cite{Mer90, Per90}, the
questions are $x,y\in[3]$.  Alice's answer $a\in\{0,1\}^3$ has even parity, Bob's answer
$b\in\{0,1\}^3$ has odd parity, and their answers must agree at the queried cell.  Thus
\begin{equation}
 \Rel_{\mathrm{MS}}(x,y)
 :=\left\{(a,b):
 \bigoplus_{j=1}^3a_j=0,\ 
 \bigoplus_{i=1}^3b_i=1,\ 
 a_y=b_x
 \right\}.
 \label{eq:magic-square-relation}
\end{equation}

For any two question pairs $(x,y)$ and $(x',y')$, the corresponding valid-output sets intersect.
Indeed, take $a=000$.  If $x\ne x'$, take $b$ to be zero in positions $x,x'$ and one in the
remaining position.  If $x=x'$, take $b_x=0$ and put one in exactly one of the other positions.
In either case $b$ has odd parity, and the same pair $(a,b)$ is valid for both question pairs.

It follows that any $g$ and $\phi$ satisfying Condition~\cref{cond:1} for
$\Rel_{\mathrm{MS}}$ must have constant $g$: a common valid output forces the two corresponding
values of $g$ to agree.  On the other hand, there is no output valid for all nine question pairs.
Such an output would satisfy $a_y=b_x$ for every $x,y$, so all six answer bits would equal a common
bit.  The even parity of $a$ would force this bit to be zero, while the odd parity of $b$ would
force it to be one.  Therefore \eqref{eq:relation-structure} cannot hold for the magic-square
winning relation with any function $g$.  Its treatment requires an argument for genuinely
relational game predicates rather than a reduction to a single function value.

\section{Proof of \cref{lem:wullschleger}}
\label{app:wullschleger}
This appendix proves \cref{lem:wullschleger} following the argument of
Ref.~\cite[Section~4]{Wul11}.  Let \(T\) be uniform on \(\binom{[N]}k\).  We seek an upper bound on
\[
 p_{\mathrm{guess}}(X_T\mid TQ).
\]
For \(S\in\binom{[N]}j\), write
\[
 \operatorname{Par}^{(N,j)}_S(x):=\bigoplus_{i\in S}x_i.
\]
The main connection between guessing \(X_T\) and guessing these parities is given by
\cref{lem:parity-to-substring}:
\[
 p_{\mathrm{guess}}(X_T\mid TQ)
 \le
 2^{1-k}\sum_{j=0}^k\binom kj
 \left(
 p_{\mathrm{guess}}
 \bigl(\operatorname{Par}^{(N,j)}_S(X)\mid SQ\bigr)-\frac12
 \right),
\]
where \(S\) is uniform on \(\binom{[N]}j\) in the term indexed by \(j\).  An intuition behind the relation is that a guess
\(\widehat X_T\) for the whole substring also gives a guess
\(\operatorname{Par}^{(N,j)}_S(\widehat X)\) for the parity of any \(S\subseteq T\);
\cref{lem:parity-to-substring} makes this relation precise.

It therefore remains to show that
\[
 p_{\mathrm{guess}}
 \bigl(\operatorname{Par}^{(N,j)}_S(X)\mid SQ\bigr)-\frac12
\]
is small for most \(j'\)s.  This is the content of \cref{prop:random-parity}.
To prove it, \cref{lem:fixed-weight-list-decoding}, which is taken from
Ref.~\cite[Theorem~44]{IJK09}, and \cref{lem:code-to-extractor} show that
\[
 \operatorname{Ext}(x,S):=\operatorname{Par}^{(N,j)}_S(x)
\]
satisfies a special property for being a strong classical bit extractor (see~\cref{def:bit-extractor}), meaning that its output is close to a uniform bit even with the conditional probability on $S$, if $H_{\min} (X)$ is large. 
Proposition~\ref{prop:KT-lifting}, which is the one-bit case of
Ref.~\cite[Theorem~III.1]{KT08}, then shows that its output remains close to a uniform bit
conditioned on \(SQ\), provided that \(H_{\min}(X\mid Q)\) is sufficiently large.  Finally, in the
sum over $j$'s above, \cref{prop:random-parity} bounds the terms with \(j>k/4\), while the sum of the
coefficients \(2^{-k}\binom kj\) for \(j\le k/4\) is exponentially small.  Combining these two
bounds proves \cref{lem:wullschleger}.

Throughout this appendix, $\log$ is base two and $\ln$ is the natural logarithm.

\subsection{Set-up and Facts}
For a classical bit $Z$ and a register $Q$, write
\[
 d(Z\mid Q)
 :=
 \frac12
 \left\|
 \rho_{ZQ}-\frac{I_Z}{2}\otimes\rho_Q
 \right\|_1,
 \qquad
 d(Z)
 :=
 \frac12
 \left\|
 \rho_Z-\frac{I_Z}{2}
 \right\|_1
\]
where $\|\cdot \|_1$ is the trace norm.
\par
When all registers are classical, trace distance is total variation distance.  In particular, if
$Z$ is a classical bit and $Y$ is classical, then
\begin{equation}
 d(Z\mid Y)
 =\E_Y\left|\Prb[Z=0\mid Y]-\frac12\right|
 =p_{\mathrm{guess}}(Z\mid Y)-\frac12.
 \label{eq:classical-bit-distance}
\end{equation}
The same operational identity holds for quantum side information:
\begin{equation}
 p_{\mathrm{guess}}(Z\mid Q)=\frac12+d(Z\mid Q)
 \label{eq:quantum-bit-distance}
\end{equation}
due to the standard characterization of the binary state-discrimination formula~\cite[Theorem 3.4]{Wat18}.

\begin{Fact}\label{Fact_App_guess_classical}
	For classical variables $X,Y$, 
$p_\mathrm{guess}(X|Y)= \E_Y[\max_x\Prb[X=x|Y]].$
\end{Fact}

\begin{Fact}\label{Fact_App_trace_classical}
	For classical variables $X,Y$ and possibly quantum system $Q$,	$d(X|YQ) = E_{y \sim Y}[d(X|Q)_{Y=y}].$
\end{Fact}

\begin{Fact}\label{Fact_App_guess_DP}
	For a classical variables $X$ with quantum side information $Q$, 
$p_\mathrm{guess}(X|Q) \ge p_\mathrm{guess}(X|Q')$ for any quantum channel $Q \to Q'$.
\end{Fact}
\subsection{Approximate list decoding and classical extraction}
Below we define \emph{strong classical bit extractors} and \emph{approximate list decodings}, and then investigate their properties and relations in \cref{lem:fixed-weight-list-decoding,lem:code-to-extractor}.
\cref{lem:fixed-weight-list-decoding} provides an existence of a ``good'' approximate list decoding,
and then, \cref{lem:code-to-extractor} shows such a good decoder may be used for a strong classical bit extractor.

Start with the following definitions.
\begin{Def}[Strong classical bit extractors]
\label{def:bit-extractor}
Let $\operatorname{Ext}:\{0,1\}^N\times\mathcal S\to\{0,1\}$, and let $R$ be uniform on
$\mathcal S$ and independent of the source.  The function $\operatorname{Ext}$ is a
\emph{strong $(K,\eta)$ classical bit extractor} if
\[
 d\bigl(\operatorname{Ext}(X,R)\mid R\bigr)\le\eta
\]
for every classical source $X$ satisfying $H_{\min}(X)\ge K$.  
\end{Def}

For $u,v\in\{0,1\}^m$, let the normalized Hamming distance be
\[
 \Delta_H(u,v):=\frac1m\bigl|\{i\in[m]:u_i\ne v_i\}\bigr| = \Prb_{I \in_U [m]}[u_I \ne v_I].
\]

\begin{Def}[Approximate list-decodable codes]
\label{def:approximate-list-decoding}
A code $C:\{0,1\}^N\to\{0,1\}^{\rev{M}}$ is
\emph{$(\eta,\delta,L)$ approximately list-decodable} if, for every
$w\in\{0,1\}^{\rev{M}}$, there is a set
$\mathcal L_w\subseteq\{0,1\}^N$ with $|\mathcal L_w|\le L$ such that, for every
$x\in\{0,1\}^N$,
\begin{equation}
 \Delta_H(w,C(x))<\frac12-\eta
 \quad\Longrightarrow\quad
 \Delta_H(x,x')\le\delta
 \text{ for some }x'\in\mathcal L_w.
 \label{eq:approx-list-definition}
\end{equation}
\end{Def}

Equivalently, for each $w$, the set of strings satisfying the left-hand side of
\eqref{eq:approx-list-definition} is contained in
\[
 \bigcup_{x'\in\mathcal L_w}
 \{x:\Delta_H(x,x')\le\delta\}.
\]

We use $[N]:=\{1,\ldots,N\}$ and write $\binom{[N]}j$ for the family of all $j$-element subsets of
$[N]$.  For $S\subseteq[N]$, define
\[
	\operatorname{Par}^{(N,j)}_S(x):=\bigoplus_{i\in S}x_i.
\]
For $1\le j\le N/2$, define the weight-$j$ parity encoding
\[
 \operatorname{Par}^{(N,j)}:\{0,1\}^N\longrightarrow
 \{0,1\}^{\binom Nj},
 \qquad
 x\longmapsto\bigl(\operatorname{Par}^{(N,j)}_S(x)\bigr)_{S\in\binom{[N]}j}.
\]

\cref{lem:fixed-weight-list-decoding} provides an existence of a good approximately list-decodable code.
\begin{Lem}[{\cite[Theorem~44]{IJK09}}]
\label{lem:fixed-weight-list-decoding}
There are universal constants $A,B\ge1$ with the following property.  Let
$1\le j\le N/2$, $0<\theta<1/2$, and $0<\delta\le1/2$.  If
\begin{equation}
 \delta j\ge A\ln\frac1\theta,
 \label{eq:fixed-weight-list-condition}
\end{equation}
then $\operatorname{Par}^{(N,j)}$ is $(\theta,\delta,B/\theta^2)$ approximately list-decodable.
Equivalently, for every word $w\in\{0,1\}^{\binom Nj}$, all strings $x$ satisfying
\[
 \Prb_{S\sim\binom{[N]}j}
 \left[w_S=\operatorname{Par}^{(N,j)}_S(x)\right]>\frac12+\theta
\]
are contained in the union of at most $B/\theta^2$ Hamming balls of relative radius $\delta$.
\end{Lem}

To see the equivalence explicitly, observe that
\begin{align*}
 \Delta_H\bigl(w,\operatorname{Par}^{(N,j)}(x)\bigr)
 &=\Prb_{S\sim\binom{[N]}j}
   \left[w_S\ne\operatorname{Par}^{(N,j)}_S(x)\right],\\
 1-\Delta_H\bigl(w,\operatorname{Par}^{(N,j)}(x)\bigr)
 &=\Prb_{S\sim\binom{[N]}j}
   \left[w_S=\operatorname{Par}^{(N,j)}_S(x)\right].
\end{align*}
Hence distance less than $1/2-\theta$ is the same as agreement probability greater than
$1/2+\theta$.

\cref{lem:code-to-extractor} shows that good approximately list-decodable codes may be used for strong classical bit extractors.
\begin{Lem}
\label{lem:code-to-extractor}
Let $0<\eta<1/2$.  Suppose that $C:\{0,1\}^N\to\{0,1\}^{\rev{M}}$ is
$(\eta/2,\delta,L)$ approximately list-decodable, where $0\le\delta\le1/2$.  If $R$ is uniform on
$[\rev{M}]$, so that it selects a uniformly random coordinate of the codeword, and
\[
 \operatorname{Ext}_C(x,R):=C(x)_R,
\]
then $\operatorname{Ext}_C$ is a strong $(K,\eta)$ classical bit extractor for
\begin{equation}
 K>H_2(\delta)N+\log L+\log\frac2\eta.
 \label{eq:code-extractor-threshold}
\end{equation}
Here $H_2(u):=-u\log u-(1-u)\log(1-u)$ is the binary entropy.
\end{Lem}

\begin{proof}
We first use the standard Hamming-ball estimate: for any $x \in \Bset^N$,
\begin{equation}
\bigl|\{x': \Delta_H(x, x')\le \delta \}\bigr|
 = \sum_{i=0}^{\lfloor\delta N\rfloor}\binom Ni\le2^{H_2(\delta)N}
 \qquad(0\le\delta\le1/2).
 \label{eq:hamming-ball-volume}
\end{equation}
Indeed, when $0<\delta\le1/2$ and $i\le\delta N$, we have
\[
 \delta^i(1-\delta)^{N-i}
 \ge\delta^{\delta N}(1-\delta)^{(1-\delta)N}.
\]
Consequently,
\[
 1
 \ge\sum_{i=0}^{\lfloor\delta N\rfloor}
     \binom Ni\delta^i(1-\delta)^{N-i}
 \ge\delta^{\delta N}(1-\delta)^{(1-\delta)N}
     \sum_{i=0}^{\lfloor\delta N\rfloor}\binom Ni.
\]
Rearranging proves
\eqref{eq:hamming-ball-volume}; the case $\delta=0$ is immediate.

Suppose now that the extractor conclusion fails for a source $X$ satisfying
$H_{\min}(X)\ge K$.  For each $r$, let $g(r)$ be a most likely value of $C(X)_r$ which implies
\[
\Prb[g(R) = C(X)_R] = \E_R\left| \Prb[C(X)_R = 0|R] -\frac{1}{2}\right| + \frac{1}{2} 
= d(\Ext_C(X,R)\mid R) + \frac 12.
\]
By \eqref{eq:classical-bit-distance}, the assumption ``$d(\Ext_C(X,R)\mid R) > \eta$ for the source $X$'' gives the deterministic function 
$g:[\rev{M}]\to\{0,1\}$ for which
\[
 \Prb[g(R)=C(X)_R]>\frac12+\eta.
\]
Define
\[
 \mathcal D
 :=\left\{x:
 \Prb_R[g(R)=C(x)_R]>\frac12+\frac\eta2
 \right\}.
\]
We show both $\Prb[X \in \mathcal D] > \eta/2$ and $\Prb[X \in \mathcal D] < \eta/2$ to reach the contradiction.

\paragraph{The lower bound $\Prb[X \in \mathcal D] > \eta/2$.}
Bounding the success probability by one on $\mathcal D$ and by
$1/2+\eta/2$ outside $\mathcal D$, we obtain
\[
 \frac12+\eta
 <\Prb[g(R)=C(X)_R]
 \le\Prb[X\in\mathcal D]
   +\bigl(1-\Prb[X\in\mathcal D]\bigr)
    \left(\frac12+\frac\eta2\right).
\]
It follows that
$\Prb[X\in\mathcal D]>\eta/(1-\eta)>\eta/2$.

\paragraph{\rev{The upper bound $\Prb[X \in \mathcal D] < \eta/2$.}}
Regard $w=(g(1),\ldots,g(\rev{M}))$ as a received word.  For every
$x\in\mathcal D$, it agrees with the condition for the approximate list-decoding; 
\[
\Prb_{R}[w_R = C(x)_R] > 1/2 + \eta/2.
\]
Hence the definition of approximate list-decoding in Definition~\ref{def:approximate-list-decoding} places $\mathcal D$ in the union of at most $L$ normalized Hamming balls of radius $\delta$:
\[
\mathcal D \subseteq \bigcup_{1 \le j \le L} \{x': \Delta_H(x_j, x')\le \delta \}
\qquad
\text{for some $x_j$'s,}
\]
which reads
 $|\mathcal D|\le L2^{H_2(\delta)N}$ by the volume bound~\eqref{eq:hamming-ball-volume}.
 The min-entropy assumption means $\Prb[X=x]\le2^{-K}$ for every $x$.  Consequently,
\[
 \Prb[X\in\mathcal D]
 \le|\mathcal D|2^{-K}
 \le L2^{H_2(\delta)N-K}
 <\frac\eta2,
\]
where the last inequality is precisely \eqref{eq:code-extractor-threshold}.  
We therefore obtain the contradiction. 
\end{proof}

\subsection{Lifting to quantum side information}
Here we investigate how properties of classical bit extractors survive when quantum side information exists.
Proposition~\ref{prop:KT-lifting} below shows a classical bit extractor behaves nicely even with the side information, and is the one-bit specialization of
Ref.~\cite[Theorem~III.1]{KT08}.  Its proof uses a property of pretty-good measurement.
 We first establish 
in \cref{lem:binary-pgm-lifting}
 the property that
converts a classical guessing probability of a binary random variable into the quantum one.

\begin{Lem}
\label{lem:binary-pgm-lifting}
Let $\rho_{ZQ} = \sum_{z \in \Bset}|z \rangle \langle z|\otimes \sigma_z$ be a cq-state.
Let $\widehat Z$ be the outcome of the pretty-good measurement $\mathbf{M}^\mathsf{PGM}$ with effects
\[
 \left\{M^\mathsf{PGM}_z := (\sigma_0+\sigma_1)^{-1/2}\sigma_z(\sigma_0+\sigma_1)^{-1/2}
 \right\}_{z \in \Bset},
\]
where the inverse is taken on $\operatorname{supp}(\sigma_0+\sigma_1)$.
Then
\[
 d(Z\mid Q)
 \le \sqrt{\frac{d(Z\mid\widehat Z)}{2}}+d(Z).
\]
\end{Lem}

\begin{proof}
Write $\sigma:=\sigma_0+\sigma_1$ and $\Delta:=\sigma_0-\sigma_1$.
Then
\begin{align}
\begin{split}
 d(Z\mid Q)&=\frac12\lVert\frac 12 \left\{| 0\rangle \langle 0| \otimes \Delta -|1 \rangle \langle 1| \otimes \Delta \right\} \rVert_1 = \frac12\lVert\Delta\rVert_1,\\
 d(Z)&=\frac12 \left\{|\Tr \sigma_0 - \frac12| + |\Tr \sigma_1 - \frac12|\right\} = \frac12\bigl|\operatorname{Tr}\Delta\bigr|.
\end{split}
 \label{eq:binary-pgm-block}
\end{align}

Since $\operatorname{Tr}\Delta=\Prb[Z=0]-\Prb[Z=1]$, subtracting
$(\operatorname{Tr}\Delta)\sigma$ removes the component arising solely from the marginal bias of
$Z$.  The factorization
\[
 \Delta-(\operatorname{Tr}\Delta)\sigma
 =\sigma^{1/4}
 \bigl[\sigma^{-1/4}\Delta\sigma^{-1/4}
       -(\operatorname{Tr}\Delta)\sigma^{1/2}\bigr]
 \sigma^{1/4}
\]
together with two applications of the Cauchy--Schwarz inequality for the Hilbert--Schmidt inner
product\footnote{
	Take $A := \sigma^{1/4}, B:= A^{-1}\Delta A^{-1} - (\Tr \Delta) A^{2}, U_\Mopt := \mathsf{sgn}(ABA), P:= A^2$ and use
\begin{align*}
\bigl\lVert\Delta-(\operatorname{Tr}\Delta)\sigma\bigr\rVert_1
 &= \|ABA\|_1 =|\Tr U_\Mopt ABA| = |\Tr AU_\Mopt AB| \le \|AU_\Mopt A\|_2 \|B\|_2,\\
\|AU_\Mopt A\|_2^2 &= |\Tr U_\Mopt PU_\Mopt P| \le \|U_\Mopt PU_\Mopt\|_2 \|P\|_2 = \|P\|_2^2 = \Tr \sigma =1.
\end{align*}
} gives
\begin{align*}
 \bigl\lVert\Delta-(\operatorname{Tr}\Delta)\sigma\bigr\rVert_1^2
 &\le\bigl\lVert\sigma^{-1/4}\Delta\sigma^{-1/4}
       -(\operatorname{Tr}\Delta)\sigma^{1/2}\bigr\rVert_2^2\\
 &=\operatorname{Tr}\!\left(\sigma^{-1/2}\Delta
       \sigma^{-1/2}\Delta\right)-(\operatorname{Tr}\Delta)^2.
\end{align*}

For the equality in the second line, expanding the square without assuming that $\sigma$ and
$\Delta$ commute gives
\begin{align*}
 \bigl\lVert\sigma^{-1/4}\Delta\sigma^{-1/4}
       -(\operatorname{Tr}\Delta)\sigma^{1/2}\bigr\rVert_2^2
 &= \Tr\bigl(\sigma^{-1/4}\Delta\sigma^{-1/4}
       -(\operatorname{Tr}\Delta)\sigma^{1/2}\bigr)^2\\
 &=\operatorname{Tr}\!\left[
       \bigl(\sigma^{-1/4}\Delta\sigma^{-1/4}\bigr)^2\right]
 -(\operatorname{Tr}\Delta)
       \operatorname{Tr}\!\left(\sigma^{-1/4}\Delta\sigma^{1/4}\right)\\
 &\qquad
 -(\operatorname{Tr}\Delta)
       \operatorname{Tr}\!\left(\sigma^{1/4}\Delta\sigma^{-1/4}\right)
 +(\operatorname{Tr}\Delta)^2\operatorname{Tr}\sigma.
\end{align*}
By cyclicity of trace, the first term is
$\operatorname{Tr}(\sigma^{-1/2}\Delta\sigma^{-1/2}\Delta)$.
Moreover, $0\le\sigma_z\le\sigma$ implies that $\sigma_z$, and hence $\Delta$, is supported on
$\operatorname{supp}(\sigma)$.  Since $\sigma^{1/4}\sigma^{-1/4}$ is the projection onto this
support, each of the two middle traces is $\operatorname{Tr}\Delta$.  Finally,
$\operatorname{Tr}\sigma=1$, so the last three terms sum to $-(\operatorname{Tr}\Delta)^2$, as
claimed.

The difference of the two measurement effects is
$\sigma^{-1/2}\Delta\sigma^{-1/2}$.  Therefore
\begin{align*}
 \operatorname{Tr}\!\left(\sigma^{-1/2}\Delta\sigma^{-1/2}\Delta\right)
 &=\operatorname{Tr}\!\left[
   \bigl(M^{\mathsf{PGM}}_0-M^{\mathsf{PGM}}_1\bigr)(\sigma_0-\sigma_1)
   \right]\\
 &=\operatorname{Tr}(M^{\mathsf{PGM}}_0\sigma_0)
   +\operatorname{Tr}(M^{\mathsf{PGM}}_1\sigma_1)
   -\operatorname{Tr}(M^{\mathsf{PGM}}_1\sigma_0)
   -\operatorname{Tr}(M^{\mathsf{PGM}}_0\sigma_1)\\
 &=\Prb[\widehat Z=Z]-\Prb[\widehat Z\ne Z]
 =2 \Prb[\widehat Z=Z] - 1\\
 &\le2 \left(p_\mathrm{guess}(Z|\widehat Z) - \frac 12 \right)
= 2d(Z\mid\widehat Z).
\end{align*}
The inequality holds because the raw outcome $\widehat Z$ is one possible decision rule, whereas
$p_\mathrm{guess}$ is the optimal guessing probability. 
In summary, it holds that 
\begin{align*}
 \bigl\lVert\Delta-(\operatorname{Tr}\Delta)\sigma\bigr\rVert_1^2
 &\le \operatorname{Tr}\!\left(\sigma^{-1/2}\Delta
       \sigma^{-1/2}\Delta\right)-(\operatorname{Tr}\Delta)^2.\\
 &\le \operatorname{Tr}\!\left(\sigma^{-1/2}\Delta\sigma^{-1/2}\Delta\right)\\
 &\le 2d(Z\mid\widehat Z).
\end{align*}

The triangle inequality and the expression~\eqref{eq:binary-pgm-block} now give the desired bound
\begin{align*}
 d(Z\mid Q)
 = \frac 12 \|\Delta\|_1
 \le \frac 12 \left( \|\Delta - (\Tr\Delta)\sigma\|_1 + \|(\Tr\Delta)\sigma\|_1\right)
 &\le\frac12
   \bigl\lVert\Delta-(\operatorname{Tr}\Delta)\sigma\bigr\rVert_1
   +d(Z)\\
 &\le\sqrt{\frac{d(Z\mid\widehat Z)}2}+d(Z).
\end{align*}
\end{proof}

\begin{Prop}
\label{prop:KT-lifting}
Let $0<\eta<1$.  If $\operatorname{Ext}$ is a strong $(K,\eta)$ classical bit extractor, then, for
every cq-state
$\rho_{XQ}$ satisfying
\[
 H_{\min}(X\mid Q)\ge K+\log\frac1\eta,
\]
and every independent uniform randomness $R$,
\[
 d\bigl(\operatorname{Ext}(X,R)\mid RQ\bigr)\le2\sqrt\eta.
\]
\end{Prop}

\begin{proof}
Let $\widehat X$ be the outcome of the pretty-good measurement for $X$ given $Q$, whose effects are
\[
 \left\{
 \rho_Q^{-1/2}\langle x|\rho_{XQ}|x\rangle\rho_Q^{-1/2}
 \right\}_x.
\]
The inverse is taken on $\operatorname{supp}(\rho_Q)$, on which these effects sum to the identity.
This measurement is independent of $R$.  Moreover, data processing~(\cref{Fact_App_guess_DP}) and the entropy assumption give
\begin{equation}
 p_{\mathrm{guess}}(X\mid\widehat X)
 \le p_{\mathrm{guess}}(X\mid Q)
 \le\eta2^{-K}.
 \label{eq:Prop_bound_classical}
\end{equation}
Applying Markov's inequality to the nonnegative random variable
\[
 \max_x\Prb[X=x\mid\widehat X],
\]
whose expectation is $p_{\mathrm{guess}}(X\mid\widehat X)$ as in \cref{Fact_App_guess_classical}, yields
\begin{equation}
 \Prb\left[
   \max_x\Prb[X=x\mid\widehat X]>2^{-K}
 \right]\le\eta.
 \label{eq:pgm-bad-outcomes}
\end{equation}
For the other outcome $\hat x$ satisfying $\max_x\Prb[X=x\mid\widehat X= \hat x] \le 2^{-K}$, the conditional source has min-entropy at least $K$, and $R$ remains
uniform and independent. Averaging the value $d\bigl(\operatorname{Ext}(X,R)\mid R\widehat X\bigr)$ over these outcomes and using the
bound $d\le1$ on the exceptional outcomes in \eqref{eq:pgm-bad-outcomes} together with the extractor guarantee, we obtain
\begin{align}
 d\bigl(\operatorname{Ext}(X,R)\mid R\widehat X\bigr)
 &= \E_{\widehat X} [d(\operatorname{Ext}(X,R)\mid R)_{\widehat X}]\notag \\ 
 &\le \Prb[\widehat X \text{~not exceptional}] \cdot \eta + \Prb[\widehat X \text{~exceptional}]\notag \\ 
 &\le 2\eta
 \label{align:fixed-measurement-extractor}
\end{align}

For a fixed instance $r$, postprocessing $\widehat X$ to
$\operatorname{Ext}(\widehat X,r)$ coarsens this measurement to the two effects
\[
 \sum_{x:\operatorname{Ext}(x,r)=z}
 \rho_Q^{-1/2}\langle x|\rho_{XQ}|x\rangle\rho_Q^{-1/2},
 \qquad z\in\{0,1\}.
\]
Because the subnormalized states belonging to the same extractor output add, these are precisely
the binary pretty-good measurement effects for $\operatorname{Ext}(X,r)$.  Furthermore,
\eqref{align:fixed-measurement-extractor} and data processing~(\cref{Fact_App_guess_DP}) imply
\[
 d\bigl(\operatorname{Ext}(X,R)\mid
        R\operatorname{Ext}(\widehat X,R)\bigr)
 \le2\eta.
\]
Moreover, the bound \eqref{eq:Prop_bound_classical} on $p_{\mathrm{guess}}(X\mid Q)$ tells 
\[
 p_{\mathrm{guess}}(X)
 \le p_{\mathrm{guess}}(X\mid Q)
 \le\eta2^{-K}
 \le2^{-K}.
\]
Thus $H_{\min}(X)\ge K$, and the property of the extractor $\operatorname{Ext}$ also gives
\[
 d\bigl(\operatorname{Ext}(X,R)\mid R\bigr)\le\eta.
\]
Applying \cref{Fact_App_trace_classical} and Lemma~\ref{lem:binary-pgm-lifting} for each randomness $r$ and then Jensen's inequality, we conclude
that
\begin{align*}
 d\bigl(\operatorname{Ext}(X,R)\mid RQ\bigr)
 &\le\sqrt{\frac{d\bigl(\operatorname{Ext}(X,R)\mid
 R\operatorname{Ext}(\widehat X,R)\bigr)}{2}}
        +d\bigl(\operatorname{Ext}(X,R)\mid R\bigr)\\
 &\le\sqrt\eta+\eta
 \le2\sqrt\eta.
\end{align*}
\end{proof}

\subsection{From random parities $\operatorname{Par}^{(N,j)}_S(X)$ to a random substring $X_T$}

Let $X\in\{0,1\}^N$ be the classical part of a cq-state $\rho_{XQ}$.  For $0\le j\le N$, let
$S$ be uniform on $\binom{[N]}j$.  By \eqref{eq:quantum-bit-distance},
\begin{equation}
 p_{\mathrm{guess}}\bigl(\operatorname{Par}^{(N,j)}_S(X)\mid SQ\bigr)
 =\frac12+d\bigl(\operatorname{Par}^{(N,j)}_S(X)\mid SQ\bigr).
 \label{eq:parity-distance-identity}
\end{equation}
For $j=0$, the guessing probability on the left is one.

We first upper-bound the guessing probability $p_{\mathrm{guess}}\bigl(\operatorname{Par}^{(N,j)}_S(X)\mid SQ\bigr)$ as follows.
\begin{Prop}
\label{prop:random-parity}
Let $0<\eta<1/2$, let $1\le j\le N/2$, and put
\[
 \delta_j:=\frac{A}{j}\ln\frac2\eta,
\]
where $A$ and $B$ are the constants from Lemma~\ref{lem:fixed-weight-list-decoding}.  Suppose that
$\delta_j\le1/2$.  If
\begin{equation}
 H_{\min}(X\mid Q)
 >H_2(\delta_j)N+4\log\frac1\eta+\log(8B),
 \label{eq:random-parity-entropy}
\end{equation}
then
\[
 p_{\mathrm{guess}}\bigl(\operatorname{Par}^{(N,j)}_S(X)\mid SQ\bigr)
 \le\frac12+2\sqrt\eta.
\]
\end{Prop}

\begin{proof}
Apply Lemma~\ref{lem:fixed-weight-list-decoding} with $\theta=\eta/2$.  The radius is chosen to be $\delta_j$
because
\[
 \delta_j\,j=A\ln\frac2\eta=A\ln\frac1\theta,
\]
so \eqref{eq:fixed-weight-list-condition} holds with equality.  The list size is at most
\[
 \frac{B}{(\eta/2)^2}=\frac{4B}{\eta^2}.
\]
Thus Lemma~\ref{lem:code-to-extractor} shows that the parity map is a strong $(K,\eta)$ classical
bit extractor for every
\[
 K>H_2(\delta_j)N+\log\frac{4B}{\eta^2}+\log\frac2\eta
  =H_2(\delta_j)N+3\log\frac1\eta+\log(8B).
\]
Indeed, for the encoding $\operatorname{Par}^{(N,j)}$, the corresponding extractor is
$\operatorname{Ext}(x,S)=\operatorname{Par}^{(N,j)}_S(x)$.
Assumption~\eqref{eq:random-parity-entropy} allows us to choose such a $K$ with
$H_{\min}(X\mid Q)\ge K+\log(1/\eta)$.  Proposition~\ref{prop:KT-lifting} therefore gives
\[
 d(\operatorname{Par}^{(N,j)}_S(X)\mid SQ)\le2\sqrt\eta.
\]
The identity \eqref{eq:parity-distance-identity} now gives the claimed bound on
$p_{\mathrm{guess}}(\operatorname{Par}^{(N,j)}_S(X)\mid SQ)$.
\end{proof}

\begin{Lem}
\label{lem:parity-to-substring}
Let $T$ be a uniformly random $k$-subset of $[N]$.  Then
\begin{equation}
 p_{\mathrm{guess}}(X_T\mid TQ)
 \le2^{-k}\sum_{j=0}^k\binom kj
 \left(2p_{\mathrm{guess}}\bigl(\operatorname{Par}^{(N,j)}_S(X)\mid SQ\bigr)-1\right),
 \label{eq:parity-to-substring}
\end{equation}
where, in the summand indexed by $j$, the set $S$ is uniform on $\binom{[N]}j$.
\end{Lem}

\begin{proof}
For each $t\in\binom{[N]}k$, fix an optimal measurement for guessing $X_t$.  When the random set
$T$ is selected, denote the output of the corresponding measurement by $\widehat X_T$ and define
\[
 W_T:=\widehat X_T\oplus X_T\in\{0,1\}^T.
\]
Thus the decoder succeeds exactly when $W_T=0$.  For each fixed $T=t$, multiplying the identities
\[
 \mathbf{1}\{W_{T,i}=0\}=\frac{1+(-1)^{W_{T,i}}}{2},
 \qquad i\in t,
\]
and expanding the product over $i$'s gives
\[
 \mathbf{1}\{W_T=0\}
 =2^{-k}\sum_{s\subseteq t}(-1)^{\oplus_{i\in s}W_{T,i}}.
\]
After choosing $T$ uniformly, choose $S$ uniformly among its $2^k$ subsets.  Taking expectations in
the preceding identity gives
\begin{equation}
 p_{\mathrm{guess}}(X_T\mid TQ)
 =\E_{T,S}\left[(-1)^{\oplus_{i\in S}W_{T,i}}\right].
 \label{eq:substring-fourier-average}
\end{equation}
The expectation also includes the source and the decoder's measurement outcome.
Equivalently, each element of $T$ is included in $S$ independently with probability $1/2$.
Therefore
\[
 \Prb[|S|=j]=2^{-k}\binom kj.
\]
Conditional on $|S|=j$, the set $S$ is uniform on $\binom{[N]}j$\footnote{  
One way to verify this is to
note that, for every fixed $s\in\binom{[N]}j$,
\[
 \Prb[S=s]
 =\frac{\binom{N-j}{k-j}}{\binom Nk\,2^k},
\]
which depends only on $j$, not on $s$.}.

To relate $p_{\mathrm{guess}}(X_T\mid TQ)$ and $p_{\mathrm{guess}}(\operatorname{Par}^{(N,j)}_S(X)\mid SQ\bigr)$ as in~\eqref{eq:parity-to-substring},
we construct a guessing estimator for $\operatorname{Par}^{(N,j)}_S(X)$ based on $\widehat X_T$ in the following way:
\begin{enumerate}
	\item Condition on $|S|=j$ and reveal $S$.
	\item Sample $T$ according to its conditional distribution given $S\subseteq T$.
	\item Apply the decoder for $X_T$.
 	\item Output the parity of the coordinates of $\widehat X_T$ indexed by $S$.
\end{enumerate}
This is a valid strategy for guessing
$\operatorname{Par}^{(N,j)}_S(X)$; the random choice of $T$ can be incorporated into the measurement.  Its
success probability is
\[
 \frac12\left(
  1+\E\left[(-1)^{\oplus_{i\in S}W_{T,i}}\,\middle|\,|S|=j\right]
 \right).
\]
It cannot exceed the optimal guessing probability.  Consequently,
\[
 \E\left[(-1)^{\oplus_{i\in S}W_{T,i}}\,\middle|\,|S|=j\right]
 \le2p_{\mathrm{guess}}\bigl(\operatorname{Par}^{(N,j)}_S(X)\mid SQ\bigr)-1.
\]
Grouping \eqref{eq:substring-fourier-average} according to $j=|S|$ now gives
\begin{align*}
 p_{\mathrm{guess}}(X_T\mid TQ)
 &=\sum_{j=0}^k2^{-k}\binom kj
   \E\left[(-1)^{\oplus_{i\in S}W_{T,i}}\,\middle|\,|S|=j\right]\\
 &\le2^{-k}\sum_{j=0}^k\binom kj
 \left(2p_{\mathrm{guess}}\bigl(\operatorname{Par}^{(N,j)}_S(X)\mid SQ\bigr)-1\right),
\end{align*}
which is \eqref{eq:parity-to-substring}.
\end{proof}

\subsection{Proof of \cref{lem:wullschleger}}
Combining all the arguments, we prove~\cref{lem:wullschleger} as follows.
\begin{proof}[Proof of Lemma~\ref{lem:wullschleger}]
Fix $\gamma>0$.  We split the sum in Lemma~\ref{lem:parity-to-substring} at degree $k/4$.
The terms of degree at most $k/4$ have exponentially small binomial weight, while
Proposition~\ref{prop:random-parity} bounds every parity bias of larger degree.

\paragraph{Choice of parameters.}
Let $A$ and $B$ be the universal constants from
Lemma~\ref{lem:fixed-weight-list-decoding}.  Choose $\beta>0$ such that, for
$\delta_\beta:=9A\beta\ln2$,
\begin{equation}
 \delta_\beta<\frac12
 \qquad\text{and}\qquad
 H_2(\delta_\beta)<\frac\gamma3.
 \label{eq:sampling-beta-choice}
\end{equation}
Choose $0<\lambda_\gamma\le1/2$ such that
\begin{equation}
 8\beta\lambda_\gamma<\frac\gamma3.
 \label{eq:sampling-density-choice}
\end{equation}
Set
\[
 c_\gamma:=\frac12\min\left\{\frac1{8\ln2},\beta\right\}.
\]
Finally, choose an integer $k_\gamma$ sufficiently large that, for every $k\ge k_\gamma$,
\begin{equation}
 2^{-2\beta k}<\frac12,
 \qquad
 \frac{4A\ln2}{k}\le A\beta\ln2,
 \qquad
 \log(8B)<\frac{\gamma k}{3},
 \qquad
 5\cdot2^{-2c_\gamma k}\le2^{-c_\gamma k}.
 \label{eq:sampling-threshold-choice}
\end{equation}
For any $N$ and $k$ satisfying $k_\gamma\le k\le\lambda_\gamma N$, put
\begin{equation}
 \eta_k:=2^{-2\beta k},
 \qquad
 \delta_j:=\frac Aj\ln\frac2{\eta_k}
 \quad (k/4<j\le k).
 \label{eq:sampling-remaining-parameters}
\end{equation}

\paragraph{Verification of the hypotheses.}
We verify simultaneously that Proposition~\ref{prop:random-parity} applies for every
$j$ with $k/4<j\le k$.  First, \eqref{eq:sampling-threshold-choice} gives
$0<\eta_k<1/2$.  Moreover, $j\ge1$ and
$j\le k\le\lambda_\gamma N\le N/2$.  Using $j>k/4$, we also have
\begin{align*}
 \delta_j
 &=\frac Aj\ln\frac2{\eta_k}\\
 &\le\frac{4A}{k}\bigl(\ln2+2\beta k\ln2\bigr)
 \le9A\beta\ln2
 =\delta_\beta
 <\frac12,
\end{align*}
where the penultimate inequality follows from \eqref{eq:sampling-threshold-choice}.
Since $H_2$ is increasing on $[0,1/2]$, our choices further imply
\begin{align*}
 H_2(\delta_j)N+4\log\frac1{\eta_k}+\log(8B)
 &<\frac{\gamma N}{3}+8\beta k+\frac{\gamma k}{3}\\
 &\le\frac{\gamma N}{3}
     +8\beta\lambda_\gamma N+\frac{\gamma N}{3}
 <\gamma N
 \le H_{\min}(X\mid Q).
\end{align*}
Thus all the hypotheses of Proposition~\ref{prop:random-parity} are satisfied.
The hypotheses of Lemma~\ref{lem:parity-to-substring} are also satisfied because $T$ is, by
assumption, uniform on $\binom{[N]}k$.

\paragraph{Application of the bounds.}
Proposition~\ref{prop:random-parity} gives, for every $j>k/4$,
\[
 2p_{\mathrm{guess}}\bigl(\operatorname{Par}^{(N,j)}_S(X)\mid SQ\bigr)-1
 \le4\sqrt{\eta_k}
 =4\cdot 2^{-\beta k}.
\]
For $j\le k/4$, we use the trivial upper bound $1$ on the same bias.
Lemma~\ref{lem:parity-to-substring} and the standard Chernoff bound therefore yield
\[
 p_{\mathrm{guess}}(X_T\mid TQ)
 \le \Prb[\operatorname{Bin}(k,1/2)\le k/4]+4 \cdot 2^{-\beta k}
 \le e^{-k/8}+4\cdot 2^{-\beta k}.
\]
By the definitions of $c_\gamma$ and $k_\gamma$, the last expression is at most
$5\cdot2^{-2c_\gamma k}\le2^{-c_\gamma k}$, as required.
\end{proof}

\section{An alternative proof of~\cref{lem:functional-sampling}}
\label{app:simplified-proof}

We give a direct proof of \cref{lem:functional-sampling} using the
pretty-good measurement and a bound on the size of a Hamming ball.
The argument also applies when $X^n$ has an arbitrary distribution,
provided that $Y^n$ is uniform and independent of $X^nU$.

Let us first recall several notations. Let $\mathcal X,\mathcal Y,\mathcal W$ be finite nonempty sets and let $g:\mathcal X\times\mathcal Y\to\mathcal W$ have pairwise distinct rows:
\[
 a\ne a'\quad\Longrightarrow\quad
 g(a,b)\ne g(a',b)\text{ for some }b\in\mathcal Y.
\]
Let $X^n$ have an arbitrary distribution on $\mathcal X^n$, with a
finite-dimensional quantum register $U$. Write their joint state as
\[
 \rho_{X^nU}=\sum_{x\in\mathcal X^n}|x\rangle\langle x|\otimes\tau_x,
 \qquad \rho_U=\sum_x\tau_x,
 \qquad \Tr\tau_x=\Prb[X^n=x].
\]
Thus the operators $\tau_x$ are positive semidefinite and
$\Tr\rho_U=1$. Define
\[
 p_{\mathrm{guess}}(X^n\mid U):=
 \max_{\{F_x\}}\sum_x\Tr(F_x\tau_x),
 \qquad
 H_{\min}(X^n\mid U):=-\log_2p_{\mathrm{guess}}(X^n\mid U),
\]
where the maximum ranges over POVMs on $U$ with outcomes in $\mathcal X^n$.

Bob receives $Y^n$, uniform on $\mathcal Y^n$ and independent of $X^nU$.
For each $y$, he applies a POVM $\{M_w^y\}_{w\in\mathcal W^n}$ to $U$.
His success probability is
\[
 \AvSP_n:=\frac1{|\mathcal Y|^n}\sum_{y\in\mathcal Y^n}\sum_{x\in\mathcal X^n}
 \Tr\bigl(M_{g^n(x,y)}^y\tau_x\bigr).
\]
Let $\Delta_H$ be the normalized Hamming distance and denote the size of radius-$r$ Hamming ball in $\mathcal X^n$ as
\[
	B_\mathcal{X}(n,r):= |\{x \mid \Delta_H(x_0,x) \le r/n\}| = \sum_{j=0}^r\binom nj(|\mathcal X|-1)^j
 \qquad(x_0 \in \mathcal X^n,~~0\le r\le n).
\]

\begin{Prop}\label{prop:direct-main}
	Assume $|\mathcal X| \ge 2$.
For every integer $0\le r<n$,
\begin{equation}\label{eq:direct-main}
 \AvSP_n^2\le
 B_\mathcal{X}(n,r)\,2^{-H_{\min}(X^n\mid U)}
 +\exp\!\left(-\frac{r+1}{|\mathcal Y|}\right).
\end{equation}
Consequently, for every fixed $h>0$, there is $\kappa=\kappa(\mathcal X,\mathcal Y,h)>0$
such that $H_{\min}(X^n\mid U)\ge hn$ implies
$\AvSP_n\le2^{-\kappa n}$ for all sufficiently large $n$.
\end{Prop}

We use the following standard inequality for the pretty-good measurement
(Ref.~\cite[Theorem~3.10]{Wat18}) and include its proof for completeness.
Throughout, inverse powers are taken on the support of the relevant state;
measurements may be completed arbitrarily on its orthogonal complement.

\begin{Lem}\label{lem:direct-pgm}
Let $\tau_z\ge0$, $\rho:=\sum_z\tau_z$, and $\Tr\rho=1$.
The POVM $E_z^{\mathrm{PGM}}:=\rho^{-1/2}\tau_z\rho^{-1/2}$ satisfies, for every POVM
$\{M_z\}_z$,
\begin{equation}\label{eq:direct-pgm}
 \left(\sum_z\Tr(M_z\tau_z)\right)^2
 \le\sum_z\Tr(E_z^{\mathrm{PGM}}\tau_z).
\end{equation}
\end{Lem}

\begin{proof}[Proof of Lemma~\ref{lem:direct-pgm}]
Work on the support of $\rho$ and put
\[
 A_z:=\rho^{-1/4}\tau_z\rho^{-1/4},
 \qquad C_z:=\rho^{1/4}M_z\rho^{1/4}.
\]
Cyclicity of trace gives $\Tr(A_zC_z)=\Tr(M_z\tau_z)$.
Hilbert--Schmidt Cauchy--Schwarz, applied to the families $(A_z)_z$ and
$(C_z)_z$, implies
\[
 \left(\sum_z\Tr(M_z\tau_z)\right)^2
 \le\left(\sum_z\Tr A_z^2\right)
     \left(\sum_z\Tr C_z^2\right).
\]
Since $0\le M_z\le I$,
\begin{align*}
 \sum_z\Tr C_z^2
 &=\sum_z\Tr\bigl(\rho^{1/2}M_z\rho^{1/2}M_z\bigr)\\
 &\le\sum_z\Tr\bigl(\rho^{1/2}M_z\rho^{1/2}\bigr)
 =\Tr\rho=1.
\end{align*}
The inequality follows from $\rho^{1/2}M_z\rho^{1/2}\ge0$ and
$I-M_z\ge0$. Finally,
\[
 \Tr A_z^2
 =\Tr\bigl(\rho^{-1/2}\tau_z\rho^{-1/2}\tau_z\bigr)
 =\Tr(E_z^{\mathrm{PGM}}\tau_z),
\]
which establishes \eqref{eq:direct-pgm}.
\end{proof}

\begin{proof}[Proof of Proposition~\ref{prop:direct-main}]
Measure $U$ with the pretty-good measurement for $X^n$,
\[
 E_{x'}^{\mathrm{PGM}}:=\rho_U^{-1/2}\tau_{x'}\rho_U^{-1/2},
\]
and denote its outcome by $\widetilde X^n$. This measurement is independent
of $Y^n$, so the pair $(X^n,\widetilde X^n)$ is independent of $Y^n$.
We compare the original protocol with the strategy that outputs
$g^n(\widetilde X^n,Y^n)$.

For fixed $y\in\mathcal Y^n$, the ensemble for the desired output is
\[
 \tau_w^y:=\sum_{x:\,g^n(x,y)=w}\tau_x,
 \qquad \sum_w\tau_w^y=\rho_U.
\]
Its pretty-good measurement is obtained by grouping the outcomes of
$\{E_{x'}^{\mathrm{PGM}}\}$:
\[
 E_w^{y,\mathrm{PGM}}:=\rho_U^{-1/2}\tau_w^y\rho_U^{-1/2}
 =\sum_{x':\,g^n(x',y)=w}E_{x'}^{\mathrm{PGM}}.
\]
Let $s_y$ be the original protocol's success probability conditioned on
$Y^n=y$. Lemma~\ref{lem:direct-pgm}, applied to $\{\tau_w^y\}_w$, gives
\[
 s_y^2\le
 \Prb\bigl[g^n(\widetilde X^n,y)=g^n(X^n,y)\bigr].
\]
Averaging over $y$ and applying Jensen's inequality, we obtain
\begin{equation}\label{eq:direct-comparison}
 \AvSP_n^2\le\E_y[s_y^2]
 \le\Prb\bigl[g^n(\widetilde X^n,Y^n)=g^n(X^n,Y^n)\bigr].
\end{equation}

We now split the set $(X^n, \widetilde X^n)$ into the two parts: 
$\Delta_H(X^n, \widetilde X^n) \le r/n$ and $\Delta_H(X^n, \widetilde X^n) > r/n$,
and we then bound $\Prb\bigl[g^n(\widetilde X^n,Y^n)=g^n(X^n,Y^n)\bigr]$ in each case separately.

\emph{Close input pairs.}
After observing $\widetilde X^n$, guess a uniformly random string in its
radius-$r$ Hamming ball. This is a valid guessing strategy for $X^n$
from $U$, with success probability
\[
	\frac{\Prb[\Delta_H(X^n,\widetilde X^n)\le r/n]}{B_\mathcal{X}(n,r)}.
\]
Therefore,
\begin{equation}\label{eq:direct-close}
 \Prb[\Delta_H(X^n,\widetilde X^n)\le r/n]
 \le B_\mathcal{X}(n,r)\,p_{\mathrm{guess}}(X^n\mid U).
\end{equation}

\emph{Distant input pairs.}
Distinct rows imply that, for $a\ne a'$,
\[
	\Prb_{Y\sim\mathrm{Unif}(\mathcal Y)}[g(a,Y)=g(a',Y)]\le1-\frac{1}{|\mathcal Y|}.
\]
Hence, for fixed $x,x'\in\mathcal X^n$ with $\Delta_H(x,x')>r/n$, independence of
the coordinates of $Y^n$ gives
\begin{equation}\label{eq:direct-far}
 \Prb[g^n(x,Y^n)=g^n(x',Y^n)]
 \le\left(1-\frac{1}{|\mathcal Y|}\right)^{n\Delta_H(x,x')}
\le e^{-(r+1)/|\mathcal Y|}.
\end{equation}
Since $(X^n,\widetilde X^n)$ is independent of $Y^n$, we may average
\eqref{eq:direct-far} over all pairs at distance greater than $r$. Thus
\[
 \Prb[g^n(\widetilde X^n,Y^n)=g^n(X^n,Y^n)]
 \le\Prb[\Delta_H(X^n,\widetilde X^n)\le r/n]+e^{-(r+1)/|\mathcal{Y}|}.
\]
Together with \eqref{eq:direct-comparison} and \eqref{eq:direct-close}, this proves
\eqref{eq:direct-main}.

\paragraph{Exponential decay.}
Suppose $H_{\min}(X^n\mid U)\ge hn$. 
Choose a constant $0<\delta<1/2$ such that
\[
 H_2(\delta)+\delta\log_2(|\mathcal X|-1)<h/2,
\]
and set $r=\lfloor\delta n\rfloor$.
The Hamming-ball estimate similar to Expression~\eqref{eq:hamming-ball-volume} shows
\begin{equation}\label{eq:direct-ball}
	B_\mathcal{X}(n,r)\le
 2^{n[H_2(\delta)+\delta\log_2(|\mathcal X|-1)]}.
\end{equation}
Since $r+1>\delta n$, Inequality~\eqref{eq:direct-main} now yields
\[
 \AvSP_n\le\sqrt{2^{-hn/2}+e^{-\delta n/|\mathcal Y|}}
 \le\sqrt2\,2^{-an/2},
 \qquad
 a:=\min\!\left\{\frac h2,\frac\delta{|\mathcal Y|\ln2}\right\}>0.
\]
Taking $\kappa=a/4$ proves the claim for $n\ge2/a$.
\end{proof}

The argument places no product or uniformity assumption on $X^n$, and no
restriction on the encoding into $U$. If $|\mathcal X|=1$, then
$H_{\min}(X^n\mid U)=0$, so the positive-entropy hypothesis is impossible.
In particular, applying Proposition~\ref{prop:direct-main} to the uniform input
distribution in \cref{lem:functional-sampling} proves that lemma without
using \cref{lem:wullschleger}.

\end{document}